\documentclass[11pt]{article} 

\usepackage[utf8]{inputenc}
\usepackage[pdftex,left=1in,top=1in,bottom=1in,right=1in]{geometry}
\usepackage{amsmath, amssymb, dsfont, mathrsfs,amsthm}
\usepackage{braket}
\usepackage{graphicx}
\usepackage{caption}
\usepackage{bm}
\usepackage{xcolor}
\usepackage[framemethod=tikz]{mdframed} %and thus tikz
\usepackage[colorlinks=true,citecolor=blue,linkcolor=blue]{hyperref}
\usepackage[nameinlink, noabbrev, capitalize]{cleveref}
\usepackage{physics}
\usepackage{verbatim} 
\usepackage{enumitem}
\usepackage{mathtools}
\hypersetup{
    colorlinks,
    linkcolor={red!50!black},
    citecolor={blue!50!black},
    urlcolor={blue!80!black}
}
\newcommand{\qvs}{\mathsf{QV}}

\mathchardef\mhyphen="2D
\DeclarePairedDelimiter\floor{\lfloor}{\rfloor}
\DeclarePairedDelimiter\ceil{\lceil}{\rceil}

\newcommand{\N}{\mathds{N}}

\newcommand{\F}{\mathds{F}}
\newcommand{\Z}{\mathbb{Z}}

\newcommand{\ske}{\mathsf{SKE}}
\newcommand{\rpke}{\mathsf{RPKE}}
\newcommand{\pke}{\mathsf{PKE}}

\newcommand{\dss}{\mathsf{DS}}

\newcommand{\hyb}{\mathsf{Hyb}}

\newcommand{\samp}{\leftarrow}
\newcommand{\reg}{\mathsf{R}}
\newcommand{\regi}[1]{\reg_{\mathsf{#1}}}

\newcommand{\io}{i\mathcal{O}}

\newtheorem{theorem}{Theorem}

 \newtheorem{lemma}{Lemma}

 \newtheorem{definition}{Definition}

\newcommand{\poly}{\mathsf{poly}}

\newcommand{\Ver}{\mathsf{Ver}}

\newcommand{\zo}{\{0, 1\}}

\newcommand{\negl}{\mathsf{negl}}
\newcommand{\subexp}{\mathsf{subexp}}

\newcommand{\keygen}{\mathsf{KeyGen}}

\newcommand{\cansbsp}{\mathsf{A}_{\mathsf{Can}}}

\newcommand{\adve}{\mathcal{A}}

\newcommand{\qmanon}[1]{\mathsf{PKQM-ANON}_{#1}}

\newcommand{\qmindresh}[1]{\mathsf{QM-FRESH-IND}_{#1}}

\newcommand{\qvsuntrac}[1]{\mathsf{QVS-PRIVACY}_{#1}}
\newcommand{\qmuntrac}[1]{\mathsf{PKQM-UNTRACE}_{#1}}

\newcommand{\bank}{\mathsf{Bank}}

\newcommand{\cfgame}[1]{\mathsf{PKQM-CF}_{#1}}
\newcommand{\qvsuniq}[1]{\mathsf{QVS-UNIQUE}_{#1}}

\newcommand{\sampfrm}[1]{\mathsf{SampleFullRank}(1^\lambda; #1)}
\newcommand{\tracegame}[1]{\mathsf{PKQM-TRACING}_{#1}}

\let\originalleft\left
\let\originalright\right
\renewcommand{\left}{\mathopen{}\mathclose\bgroup\originalleft}
\renewcommand{\right}{\aftergroup\egroup\originalright}
\title{Anonymous Public-Key Quantum Money and Quantum Voting}

\date{}

 \author{Alper \c{C}akan\thanks{Carnegie Mellon University. \texttt{acakan@andrew.cmu.edu}.} \and Vipul Goyal\thanks{NTT Research \& Carnegie Mellon University.  \texttt{vipul@vipulgoyal.org}} \and Takashi Yamakawa\thanks{NTT Social Informatics Laboratories, Tokyo, Japan. \texttt{takashi.yamakawa@ntt.com}.}}
\begin{document}
	
\maketitle
\begin{abstract}
     Quantum information allows us to build \emph{quantum money} schemes, where a bank can issue banknotes in the form of authenticatable quantum states that cannot be cloned or counterfeited: a user in possession of $k$ banknotes cannot produce $k + 1$ banknotes. Similar to paper banknotes, in existing quantum money schemes, a banknote consists of an unclonable quantum state and a classical serial number, signed by bank. Thus, they lack one of the most fundamental properties cryptographers look for in a currency scheme: privacy. In this work, we first further develop the formal definitions of privacy for quantum money schemes. Then, we construct the first public-key quantum money schemes that satisfy these security notions. Namely,
    \begin{itemize}
        \item Assuming existence of indistinguishability obfuscation and hardness of Learning with Errors, we construct a public-key quantum money scheme with anonymity against users \emph{and} traceability by authorities.
    \end{itemize}
    Since it is a policy choice whether authorities should be able to track banknotes or not, we also construct an \emph{untraceable} money scheme, where no one (not even the authorities) can track banknotes.
    \begin{itemize}
        \item Assuming existence of indistinguishability obfuscation and hardness of Learning with Errors, we construct a public-key quantum money scheme with untraceability.
    \end{itemize}
    Further, we show that the \emph{no-cloning principle}, a result of quantum mechanics, allows us to construct schemes, with security guarantees that are classically impossible, for a seemingly unrelated application: \emph{voting}!
    \begin{itemize}
        \item Assuming existence of indistinguishability obfuscation and hardness of Learning with Errors, we construct a universally verifiable quantum voting scheme with classical votes.
    \end{itemize}

    Finally, as a technical tool, we introduce the notion of \emph{publicly rerandomizable encryption with strong correctness}, where no adversary is able to produce a malicious ciphertext and a malicious random tape such that the ciphertext before and after rerandomization (with the malicious tape)  decrypts to different values! We believe this might be of independent interest.
    \begin{itemize}
        \item Assuming the (quantum) hardness of Learning with Errors, we construct a (post-quantum) classical \emph{publicly rerandomizable encryption scheme with strong correctness}.
    \end{itemize}
\end{abstract}

\newpage
\tableofcontents
\newpage

\section{Introduction}
The exotic nature of quantum mechanics allows us to build cryptographic primitives that were once unimaginable, or are outright impossible with classical information alone. For example, one of the most fundamental results of quantum mechanics, called \emph{the no-cloning principle}, shows that arbitrary unknown quantum states cannot be cloned. This simple principle, which provably has no counterpart in classical world since classical information can always be copied, allowed cryptographers to build applications that are impossible in a classical world. Starting with the seminal work of Wiesner \cite{Wie83} which introduced \emph{quantum money}, a plethora of work built exciting primitives based on the no-cloning principle. The examples include more realistic version of quantum money called \emph{public-key quantum money} \cite{AC12, Z19}, where any user can verify a banknote on their own without going to the central bank, or even more advanced notions such as copy-protecting software/functionalities where a user that is given some number of copies of a software cannot create more copies of it\footnote{Again, a classical software can always be copied, so this is impossible classically.} \cite{Aar09,CLLZ21,CG23}.

While quantum money is one of the most important notions in quantum cryptography, unfortunately existing schemes \cite{AC12,Z19} lack some of the most basic privacy and security guarantees that cryptographers look for in a currency scheme. In fact, in all known public-key quantum money schemes, a banknote consists of an unclonable quantum state and a classical serial number which is signed by the bank. However, this means that any party can track any banknote and learn when and where it was used simply by recording its serial number, meaning there is no privacy at all. For example, imagine a scenario where you pay a large sum of money to a merchant. If the merchant pools the data with other (adversarial) sources, it maybe feasible\footnote{In fact, even innocent amateur efforts that were built for fun, such as \href{https://www.wheresgeorge.com}{wheresgeorge.com}, have been able to track millions of paper dollar  banknotes all over the world.} for them to recover a history on many of these banknotes, potentially revealing your employer, clients and business partners, and even family members! Similarly an employer who pays you salary can potentially track where you travel to by collaborating with other sources and tracking, say, your spending at gas stations or restaurants. Indeed it is hard to imagine privacy in any aspect of your life, if all your spending can be traced. In fact, privacy and anonymity is the central focus in many cryptocurrency projects \cite{monero,Zcash}.

This state of affairs leaves open the following natural question:
\begin{quote}
\textit{Is it possible to construct a publicly verifiable quantum money scheme with privacy guarantees?}
\end{quote}
Let us emphasize that the above question is highly non-trivial for the following reason. To satisfy privacy guarantees, one needs to build a quantum money scheme where the users can create a new banknote so that the adversarial parties who have seen the banknote before will not be able to recognize it. However, such a task, while still challenging but doable for classical information (e.g. rerandomizable signatures \cite{camenisch2004signature}), seems to be at odds with the main point of quantum money: \emph{unclonability}!
While any user should be enabled to create new banknotes, somehow we also need to make sure that they cannot create $k+1$ valid banknotes if they started with $k$ banknotes because otherwise they can increase the amount of money they have at wish!

Going beyond privacy concerns, one useful property of the existing quantum money schemes is that since any party can track a banknote, in particular \emph{law enforcement} can also track a banknote. This brings us to our next natural question?
\begin{quote}
\textit{Is it possible to construct a publicly verifiable quantum money scheme with privacy guarantees against users, while still providing traceability for authorities?}
\end{quote}

We note that it is a political choice whether authorities should be able to track banknotes or not. Therefore, we also ask
\begin{quote}
\textit{Is it possible to construct a publicly verifiable quantum money scheme with privacy guarantees \emph{against everyone}, including the bank/authorities?}
\end{quote}

Finally, we observe an interesting connection between quantum money with privacy (which is impossible classically) and \emph{voting}. In both cases, we care about privacy and a security notion relating to \emph{non-increasibility}: in quantum money, users should not be able to increase their amount of money (in particular, given a single banknote, one should not be able to create two banknotes), and in voting, a single user should be able to vote once. Thus, we ask the following question:
\begin{quote}
\textit{Using quantum information, is it possible to construct voting schemes with advanced security guarantees that are not possible classically?}
\end{quote}

\subsection{Our Results}
In this work, we answer all of these open questions affirmatively.

We construct the first public-key quantum money scheme with anonymity (against users) and traceability (by the authorities). For anonymity, we require that a malicious user will not be able distinguish a banknote it has seen before from a freshly minted banknote. For tracing, we require that no malicious user can produce a banknote with a particular tag (which is hidden inside the serial numbers, except to authorities) without being given a banknote with that tag to begin with - meaning that law enforcement can perfectly track banknotes.
\begin{theorem}[Informal]
Assuming the existence of indistinguishability obfuscation ($\io$) and a \emph{publicly rerandomizable encryption scheme with strong correctness and publicly testable ciphertexts}, there exists a public-key quantum money scheme with anonymity and traceability.    
\end{theorem}
In fact, in our model, we separate \emph{authorities} into two completely independent entities: \emph{the bank} (who mints the banknotes) and \emph{the tracing authority}. Our scheme satisfies anonymity even against the bank, and satisfies unclonability (also called counterfeiting security) even against the tracing authority.

We note that previously, anonymous quantum money had been constructed only\footnote{Strictly speaking, \cite{bs21} calls their model \emph{almost-public}. They simply use an existing private-key quantum money scheme with pure states (more precisely, pseudorandom quantum states \cite{JLS18}), and an alleged banknote is compared to user's existing banknotes to verify it. This requires that the user always has more money than she can receive. From our point of view, this is not a public-key scheme.} in the private key setting \cite{mosca2010quantum,bs21,amr20}. However, aside from the impracticality of the private key setting, the anonymity notion in the private-key setting is also less meaningful: Once we are at the central bank to verify a banknote, we might as well ask them to replace our banknote with a fresh one. The previous solutions are based on \emph{Haar random states} or their computational version, \emph{pseudorandom states}; however, in the private-key setting, a trivial solution based on quantum fully homomorphic encryption also exists, where we can just encrypt banknotes and use the homomorphic encryption scheme to rerandomize them.

Going further, we construct the first public-key quantum money scheme with anonymity against all parties (including the bank and the authorities). We call this notion \emph{untraceability.}
\begin{theorem}[Informal]
Assuming the existence of $\io$, a \emph{publicly rerandomizable encryption scheme with strong correctness and publicly testable ciphertexts} and a non-interactive zero-knowledge (NIZK) argument system for $NP$, there exists a public-key quantum money scheme with \emph{untraceability} in the common \emph{random} string model.
\end{theorem}

Finally, we construct the \emph{first} voting scheme with universal verifiability (anyone in the world can verify any vote), privacy against all parties (including voting authority!) and uniqueness (i.e. no double-voting). We note that a voting scheme satisfying these three properties at the same time provably cannot exist in a classical voting scheme. Thus, the voting tokens of our scheme are quantum, but a cast vote is classical. 
\begin{theorem}[Informal]
    Assuming the existence of $\io$, a \emph{publicly rerandomizable encryption scheme with strong correctness and publicly testable ciphertexts} and a non-interactive zero-knowledge (NIZK) argument system for $NP$, there exists a quantum voting scheme with \emph{universal verifiability} and classical votes, in the common \emph{random} string model.
\end{theorem}
As discussed, such a scheme cannot exist classically. Even using quantum voting tokens, ours is the first to achieve these guarantees: there is no previous work achieving universal verifiability (or classical votes).

We note that while $\io$ is a strong assumption, all of our constructions above imply standard public-key quantum money, whose all existing constructions in the plain model also use $\io$\footnote{We note that there are some candidate constructions based on non-standard ad-hoc assumptions}. In fact, it is one of the key open questions in quantum cryptography to construct public-key quantum money without $\io$. Thus, unless a major breakthrough is achieved, our $\io$ assumption is necessary. 

To achieve our results, we also introduce the notion of \emph{publicly rerandomizable encryption with strong correctness}, where no adversary is able to produce a (malicious) ciphertext whose decryption result differs between before and after rerandomization (even with a maliciously chosen rerandomization randomness tape); and we construct such a scheme that is secure against quantum adversaries. Our schemes also satisfy public testability, where anyone can test whether a ciphertext is \emph{bad} (in the above sense, where rerandomization can lead to decrypting to a different message), such that there \emph{does not exist} a ciphertext that passes this test but decrypts to different values before/after rerandomization. We note that this notion is useful in constructions/proofs that use indistinguishability obfuscation, where merely computational hardness of finding bad ciphertexts would not be sufficient. Our results/constructions here are classical and we believe they might be of independent interest.

\begin{theorem}[informal]
    Assuming hardness of Learning with Errors (LWE) \cite{regev2009lattices}, there exits a \emph{publicly rerandomizable encryption with publicly testable ciphertexts and strong correctness}.
\end{theorem}
Thus, instantiating our constructions with our \emph{publicly rerandomizable encryption scheme with strong correctness}, and the NIZK argument system with the LWE-based (post-quantum) construction of Peikert-Shiehian \cite{peikert2019noninteractive} (which is in the common random string model), all of our quantum money and quantum voting schemes can be based on $\io$ and LWE.
\section{Technical Overview}
\subsection{Definitional Work}
We first review our model. In a quantum money scheme, we consider a bank that produces quantum banknotes that can be publicly verified, and the first security requirement is counterfeiting security (also called \emph{unclonability}). For counterfeiting security, we require that any efficient adversary that has the public verification key and $k$ banknotes produced by the bank, is not to able to produce $k + 1$ valid banknotes, for any $k$ (not a-priori bounded).

Prior work (\cite{mosca2010quantum,amr20,bs21}) has introduced the notion of anonymity for quantum money schemes (albeit in the privately verifiable money setting). In their anonymity security game (\cite{amr20,bs21}), an adversary observes some banknotes, and later he is (depending on a challenge bit) either given those banknotes in the same order or in a permuted order; and the security requirement is that he cannot tell which case it is. We note that one downfall of private-key setting (apart from impracticality) is that achieving anonymity is not as interesting: once at the central bank to verify a banknote, we might as well also ask the central bank to replace our banknote with a fresh one. 
We first move onto the publicly verifiable banknote setting (dubbed \emph{public-key quantum money} or PKQM for short), which is the ultimate goal in quantum money literature.  Further, we introduce a new security notion called \emph{fresh banknote indistinguishability security}. In this game, the adversary is given either the banknote it has seen before (in fact created itself in our model) or a freshly minted banknote. We require that the adversary cannot tell which case it is with probability better than $1/2 + \negl(\lambda)$. By a simple hybrid argument, one can see that this security notion implies the previous anonymity notions.

Going beyond anonymity, we also introduce the notion of \emph{traceability} for quantum money. While the notion of \emph{traceability} has a long history in cryptography, our work is the first to formally define it for quantum money. In this setting, we allow the bank/mint to include a \emph{tag} value in each banknote. In the security game, we require that an adversary that obtains some banknotes from the bank with various tag values is not able to produce a valid banknote with a new tag value. In fact, we require that even the number of banknotes with each tag value cannot be increased. Thus, our security notion implies that all that an adversary can do, while still outputting valid banknotes, is to permute the banknotes given to it, and drop some of them! This allows the authorities to correctly track banknotes. We will always consider traceability (by the authorities) alongside anonymity (against users), which makes the problem highly non-trivial.

Before moving onto our final security model, \emph{untraceability}, a couple of further remarks about our models is in order. First, we note that in our models, we consider two separate independent entities: the central bank/mint (who creates the banknotes) and a tracing authority. In real life, these entities can be two independent governmental organizations. We require that the bank can produce banknotes on its own (with no help from the tracing authority), while the tracing authority can trace banknotes on its own (with no help from the bank). In fact, we require that the anonymity holds against even the bank itself, who produced the banknotes in the first place! Further, we require that unclonability and tracing security also holds against the tracing authority.

Finally, we introduce the notion of \emph{untraceability} for the first time. In this setting, we require that there are no entities, including the bank or the government, that is able to track banknotes. Here, we require that fresh banknote security applies to even a malicious bank that is allowed to choose the banknotes and even the verification key maliciously. Due to our strong anonymity model, this security notion follows easily in the trusted setup model\footnote{We can imagine a trusted setup that honestly creates the quantum money scheme, but deliberately forgets the tracing key and only outputs the minting key secretly/directly to the central bank. Since our anonymity security applies even against the bank, the security follows.}. However, ad-hoc/application-based trusted structured setup models are highly undesirable. Thus, we consider \emph{untraceability} in the common \emph{random} string model, which is a much weaker and more realistic assumption. We refer the reader to \cref{sec:defn} for formal definitions.

\subsection{Anonymous and Traceable Construction}
\paragraph{\textbf{Background on Subspace State Quantum Money \cite{AC12,Z19}}} The starting point of our construction is the subspace-state based public-key quantum money scheme of Aaronson - Christiano \cite{AC12}, which originally used classical ideal oracles but was also proven secure using $\io$ by Zhandry \cite{Z19}. We note that, to date, subspace states (and their close relatives, such as subspace-coset states) are the only known public-key quantum money schemes with provable security based on standard assumptions.

A subspace-state is a state consisting of equal superposition over all elements of a subspace, that is, it is $\sum_{v \in A} \frac{1}{2^{n/4}}\ket{v} =\vcentcolon \ket{A}$, where $A$ is a linear subspace of the vector space $\F_2^n$. The subspace-state quantum money construction works as follows. The bank's secret key is simply a (post-quantum) classical signature scheme. To mint a banknote, the bank simply samples a random subspace $A$ of dimension $n/2$, and signs the obfuscated membership checking programs $\io(A), \io(A^\perp)$ for the subspace $A$ and its orthogonal complement $A^\perp$. Then, the banknote is the state $\ket{A}$, together with the progras and the signature on them. Here, we can consider the string $\io(A) || \io(A^\perp)$ as the serial number of the banknote. To verify an alleged banknote $(\ket{\psi}, \mathsf{sn} = \io(A) || \io(A^\perp), \mathsf{sig})$, one verifies the signature $\mathsf{sig}$ on the serial number $\mathsf{sn}$, then coherently runs the first program $\io(A)$ on the state $\ket{A}$, verifying that the output is $1$. After rewinding, we apply quantum Fourier transform (QFT) to the state $\ket{\psi}$, and this time coherently evaluate the second program $\io(A^\perp)$ to verify that it outputs $1$. \cite{AC12} shows that this primal/orthogonal basis (more accurately, called \emph{computational/Hadamard basis}) verification implements a projection onto $\ket{A}$, thus correctness follows. The counterfeiting security follows in two steps. First, they show that for a random subspace $A$ of dimension $n/2$; no adversary can produce $\ket{A}\otimes\ket{A}$ given $\ket{A}$ and the membership checking programs. In the general setting where any number of banknotes is in circulation, they show that the security of the bank's classical signature scheme is sufficient to reduce to the single banknote setting. After all, by signature security, any supposed banknote produced by the adversary will have to use one of the original serial numbers produced by the bank, thus the adversary cannot create its own banknotes by sampling subspace states: it has to try and clone one of the subspace-states produced by the bank, which is only produced as a single copy. Thus the reduction to $1\to2$ setting follows.

\paragraph{\textbf{Challenges for Anonymity}}
Observe that the subspace-based quantum money scheme is trivially trackable. Each valid banknote contains a serial number and the bank's signature on the serial number. Thus, much like paper banknotes, simply recording serial numbers is sufficient to track banknotes: whenever we receive and validate a banknote, we can simply search for the serial number in our database. 

An initial attempt to provide anonymity might be to draw from the classical literature, and for example to use \emph{re-randomizable signatures} \cite{camenisch2004signature}. We first note that, most of the existing constructions depend on number-theoretic assumptions and to the best of our knowledge, there are no known \emph{post-quantum} re-randomizable digital signature constructions in the plain model. Setting this issue aside, this initial attempt is still trivially broken, since while the signature is rerandomizable, the message (the serial number) will stay the same, which uniquely identifies the banknote. Going further, one might imagine a solution where the signature is rerandomized along with the message (serial number) inside, and we can imagine a scheme such that \emph{somehow} the subspace state can still be verified with this updated serial number, to preserve correctness of the scheme. While this would be challenging to achieve, it is still not sufficient. Observe that only information needed to perfectly recognize a banknote is its \emph{old} classical serial number. Imagine the following scenario. We (the adversary) have two banknotes $(\ket{A_1}, \mathsf{sn}_1, \mathsf{sig}_1), (\ket{A_2}, \mathsf{sn}_2, \mathsf{sg}_2)$. Later on, the challenger rerandomizes these banknotes, and gives us back one of the banknotes ($b \in \{1,2\})$ as $(\ket{\psi}, \mathsf{sn}', \mathsf{sig}')$. We can simply perform the computational basis/Hadarmard basis verification on $\ket{\psi}$ using the old serial number $\mathsf{sn}_1$ that we had recorded. Since this dual basis test implements a projection onto the subspace state $\ket{A_1}$, our test will accept $\ket{\psi}$ with probability $\braket{A_1}{\psi}$. In the case $b = 1$, this value equals $\braket{A_1}{A_1} = 1$. In the case $b = 2$, this value will be equal to $\braket{A_1}{A_2}$, which is exponentially small in expectation. Thus, using this simple test, we can recognize which banknote we got back almost perfectly! In fact, this test can performed \emph{coherently} by the Gentle Measurement Lemma \cite{aarlemma}, so that we do not even damage the banknote at all while doing the test!

Thus, we can see that the only way to provide anonymity is to actually change (rerandomize) the quantum state itself too! Here, a viable candidate could be using fully homomorphic encryption for quantum messages to rerandomize our quantum states, since homomorphism is known to be closely related to rerandomization in classical cryptography. However, this would completely forego the most  important property, public-verifiability, since the users cannot verify encrypted banknotes. Another idea could be to try and sign the quantum parts of the banknotes, and try to construct a signature and message rerandomizable signature scheme for quantum messages. However, it is known that signing quantum messages is impossible \cite{Alagic2021canyousignquantum}.

This brings us to the fundamental challenge in this task. We need to build a scheme where the users can create essentially a new quantum state, so that the adversarial parties who have seen the banknote before will not be able recognize it. However, at the same time we also need to make sure that these states are unclonable and that users cannot create more quantum states than what they already have! 

\paragraph{\textbf{Solution Step 1: Rerandomizing Quantum States}} First step of our solution is to build a new quantum money scheme (without anonymity), though the quantum state part of our scheme will still be essentially subspace states. In our scheme, the banknote will consist of a serial number $id$, and the quantum state will be a superposition over the elements of the subspace obtained by pseudorandomly applying a \emph{rotation-reflection-scaling} to the canonical $n/2$-dimensional subspace $\cansbsp = \mathsf{Span}(e_1, \dots, e_{n/2}) \subseteq \F_2^n$ (i.e. it is the subspace consisting of all vectors whose last $n/2$ components are $0$). More formally, the quantum part of the banknote will be the state $\ket{T\cdot \cansbsp} = \sum_{v \in \cansbsp} \frac{1}{2^{\lambda/4}} \ket{T(v)}$, where $T$ is the full rank linear mapping $\F_2^\lambda \to \F_2^\lambda$ sampled using the randomness $F(K, id)$ where $F$ is a pseudorandom function (PRF) and $id$ is the serial number of the banknote. The bank's secret key will simply be the PRF key $K$. Finally, the public key will be an obfuscated program $P$ that takes as input a vector $v \in \F_2^\lambda$, a bit $b \in \zo$ and a serial number $id$. The obfuscated program first recovers the hidden full rank linear map $T_{id}$ using the hardcoded PRF key $K$ and the input $id$. For $b = 0$, the program computes $w = T^{-1}(v)$, and verifies if its last $n/2$ entries are $0$ (that is, it checks $w \in \cansbsp$). For $b=1$, the program computes $w' = T^{\mathsf{T}}(v)$ and verifies if $w' \in \cansbsp^\perp$. We show that testing $P(\cdot, b=0, id)$ in superposition, then applying QFT and testing $P(\cdot, b=1, id)$ in superposition implements a projection onto the state $\ket{T_{id}\cdot \cansbsp}$, thus correctness is satisfied.

Now, we move onto adding anonymity to our construction. First observation is that the set of full rank linear mappings on $\F_2^n$ has a \emph{nice} structure (under composition/multiplication): it is the general linear group $\mathsf{GL}(n, \F_2)$. This means that we can re-randomize a (full-rank) map $T$ by simply multiplying with a fresh random (full-rank) map $T'$, due to group rerandomization property. However, we cannot allow the users to rerandomize using any $T'$, because the banknote would lose its consistency with the serial number-quantum state mapping induced by $F(K, \cdot)$, since the user will not be able to come up with the new serial number. In fact, in general, given $T, T', id$, it might not be possible to come up with a new serial number $id'$ such that the randomness $F(K, id')$ gives us the full rank map $T''$ satisfying $T'' = T' \cdot T$. To solve this, we include an obfuscated program $\mathsf{PReRand}$ as part of the public key that allows the user to renrandomize the banknote consistently. On input $id$, the program will first re-randomize $id$ using some $\mathsf{IdReRandomizeAlgorithm}$ to obtain $id'$, and then generate $T_{id'}$ using the randomness $F(K, id')$. Finally, it will output a canonical representation of the map $T^* = T'\cdot T^{-1}$, along with $id'$. Given $T^*$, the user can \emph{convert} its banknote $\ket{T\cdot \cansbsp}$ into $\ket{T' \cdot \cansbsp}$ by applying the mapping $T^*$ coherently/in-superposition. Note applying $T^*$ coherently is possible since (i) $T^*$ is a bijection and (ii) it is efficiently implementable in both directions $T^*, (T^*)^{-1}$ given the description $T^*$. Thus, we achieve a meaningful rerandomization algorithm while preserving correctness after rerandomization.

%\anote{bounded unclonable solution}

%\anote{only known verifiable unclonable states are $1\to2$ unclonable}
%\anote{update references to secure rerandomization}
\paragraph{\textbf{Solution Step 2: Non-intersecting/Traceable Rerandomization Cones}}
Since the quantum part of our scheme are subspace states, to prove the unclonability of our scheme, we will need to reduce to the unclonability of subspace states, which are only $1\text{-}\mathsf{copy} \to 2\text{-}\mathsf{copy}$ unclonable.
Now, consider the \emph{chain} (or more accurately, the tree/cone) of rerandomizations for a serial number $id$ that consists of the initial serial number $id$, and then (each possible) rerandomization of $id$, and then rerandomizations of those values and so on. Let us denote this as $\mathsf{C}_{id}$. Now consider these cones for each of the $k$ initial banknotes that the adversary obtains, $\mathsf{C}_{id_1}, \dots, \mathsf{C}_{id_k}$. Now, observe that if any of these cones intersect at any point, say $\mathsf{C}_{id_i}$ and $\mathsf{C}_{id_j}$ at $id^*$, in a way that an efficient adversary can find, then the adversary can simply emulate the path (i.e. the malicious rerandomization random tape choices) that leads to this intersection on the banknotes $id_i, id_j$ to obtain \emph{two} copies of the exact same state $\ket{T_{id^*}\cdot \cansbsp}$! However, we need to reduce to the $1\text{-}\mathsf{copy}\to 2\text{-}\mathsf{copy}$ unclonability of our quantum states. Thus, we need to design a serial number re-randomization algorithm where the adversaries cannot find such intersections of the rerandomization cones. 

%\anote{this parag fine and next fine }
In fact going further, we observe that while \emph{non-intersecting} (to efficient adversaries) rerandomization cones is sufficient to prevent identical-copy attacks, it is not sufficient to actually give a reduction to $1 \to 2$ unclonability and prove security. Observe that to be able to achieve a reduction to $1 \to 2$ unclonability, the reduction algorithm will need to place the subspace state $\ket{A}$ it receives from its challenger at some index $\in [k]$, and once the adversary produces $k + 1$ banknotes, the reduction will need to somehow extract two identical copies of $\ket{A}$ using these states. This means the reduction itself actually needs to mount a partial tracing attack, so that it can track back the forged banknotes produced by the adversary to find two that are \emph{rooted} at the same initial subspace state, and undo the rerandomizations and convert these forged banknotes to two copies of $\ket{A}$.

To resolve this issue, we show that a \emph{rerandomizable encryption} scheme with \emph{strong rerandomization correctness} gives us exactly what we want. In our money scheme, the serial numbers will be rerandomizable public-key encryption (RPKE) ciphertexts (that encrypt \emph{tags}) and $\mathsf{IdReRandomizeAlgorithm}$ will simply be the rerandomization algorithm of the RPKE scheme (and the rest of the money scheme is as described above). By the rerandomization security of the RPKE scheme, we will be able to prove anonymity easily. Further, the notion of \emph{strong rerandomization correctness} gives us the \emph{tracable} non-intersecting cone property we needed to achieve unclonability. First, let us define \emph{strong rerandomization correctnes}, which is a notion we introduce for the first time and we believe it might be of independent interest. For this notion, we require that for any (malicious) ciphertext string $ct$ and (malicious) random tape $r$ chosen by an efficient adversary, we have $\mathsf{Dec}(sk, ct) = \mathsf{Dec}(sk, \mathsf{ReRand}(pk, ct; r))$. Applying this transitively, we can see that (efficient) cones of any two serial numbers will not intersect. Finally, the decryption algorithm of the RPKE scheme gives us exactly the traceability we wanted from our rerandomization cones to do our unclonability reduction! Thus, we are able to prove unclonability and anonymity security using a \emph{rerandomizable encryption} scheme with \emph{strong rerandomization correctness}. For technical reasons relating to use of $\io$, we actually require \emph{strong correctness with public testing} - see \cref{sec:techstrongrerand}.

\subsection{Untraceable Construction}
%\anote{instead: bank/authorities can be malicious}. \anote{just say: adv cannot be trusted, programs keys public key can be malicious}
Our starting point is our anonymous and traceable quantum money scheme. The main challenging point in this setting is the fact that all parties (users and the bank) are mutually distrustful. The first issue is that, in the untraceability game, the adversary gets to choose all the values including the public key of the scheme (which includes $\mathsf{PReRand}$). Thus, the adversary can simply use the corresponding decryption key to decrypt the serial numbers and obtain the hidden tags it placed inside them, and hence track the banknotes. Further, even if we hypothetically trusted the adversary to not use the decryption key of the RPKE scheme to track banknotes, it can still create a malicious obfuscated program $\mathsf{PReRand}'$ so that the new banknote state we end up with after rerandomization includes a hidden signal that depends on the old serial number! Note that we cannot trust the user to rerandomize their serial numbers  themselves with no checks either, since that would break down our unclonability argument.

\paragraph{\textbf{Solution Step 1: Proofs of Projectiveness}} Our solution is to include, as part of the verification key, a non-interactive zero-knowledge (NIZK) argument showing that the money scheme (with its maliciously chosen keys) is still \emph{projective}. More formally, the NIZK argument will prove that there exists a (hidden) serial number - quantum state mapping $M$ such that computational/Hadamard basis verification using $\mathsf{PMem}(\cdot, b=0, id)/\mathsf{PMem}(\cdot, b=1, id)$ implements a projection onto the state $\ket{M(id)}$. Thus, any user (as part of banknote verification) can rerandomize a banknote and test again using the new $id'$, and once the test passes, the user is assured that he has the state $\ket{M(id')}$. If the test does not pass, either the initial banknote was not valid at all (i.e. a counterfeiting attack) or it was designed to be tracked (which also makes it invalid)! With the proof of projective, we are now assured that if the adversary really does not have the decryption key for RPKE, then $\ket{M(id')}$ is indistinguishable from $\ket{M(id^*)}$ where $id^*$ is a fresh ciphertext, since the indistinguishability would follow from indistinguishability of $id'$ versus $id^*$, so it still applies even though the adversary knows $M$. Finally, the zero knowledge property is needed since revealing $M$ to the users would allow them to clone states trivially. We note that post-quantum NIZKs based on LWE can be constructed in the common random string model \cite{peikert2019noninteractive}.

\paragraph{\textbf{Solution Step 2: Truly Random Public Keys}} 
However, there is still a major issue: obviously we cannot trust the adversary to not create a decryption key that corresponds to the public key. To solve this, we will simply use a truly random string as the public rerandomization key. However, we seem to be at a dead-end: we need to prove, at the same time, that no one can track banknotes, while the unclonability proof (as we discussed) needs to mount a partial tracking attack. Our solution is to construct a rerandomizable encryption scheme with \emph{strong rerandomization correctness with testing} that also satisfies (i) \emph{pseudorandom public (encryption) keys}, (ii) \emph{simulatable testing keys} (\cref{sec:techstrongrerand}) and (iii) \emph{statistical rerandomization security for truly uniformly random public keys}. This allows us to do the following: In the actual construction, the users use (part of) the common random string (CRS) as the public key of the RPKE scheme. For the unclonability proof, we rely on the pseudorandom public keys property of the RPKE scheme: The reduction simulates the cloning adversary while planting an actual public key of our RPKE scheme (which the reduction samples together with the corresponding decryption key $sk$) as the last part of the CRS, instead of a truly random string. Due to the pseudorandom public keys property of RPKE, this means that the adversary still succeeds, while the reduction can also mount the partial tracking attack using $sk$ as discussed before. Finally, we use the property (ii) because we cannot embed the ciphertext testing key inside the common uniformly random string, since such a key obviously cannot be pseudorandom. Thus, in the actual construction we use the simulated ciphertext testing key. 

For anonymity, we will be able to rely on property (iii) since we are rerandomizing the ciphertexts using a truly random string. More precisely, as discussed above, we first rely on the soundness of the NIZK argument system to show that if the adversary succeeds, there exists a mapping $M$ as defined above. However, for the challenger to start outputting $\ket{M(id')}$ directly instead of running adversary's malicious rerandomization program, we need to recover $M$. At this point, we switch to an unbounded challenger that does recover $M$ by trying all possibilities. The final step is to replace $\ket{M(id')}$ with $\ket{M(id^*)}$ relying on the rerandomization security of RPKE. Fortunately, by property (iii), we can indeed do so even with an unbounded challenger/reduction. Note that in the RPKE we end up constructing (\cref{sec:rpke}), we will only have computational rerandomization for honest public keys (which comes with a secret key), but we do have statistical rerandomization when a truly random string is used as the public key.

\subsection{Quantum Voting}
%\anote{talk a little about previous quant voting they hsve no univ ver no classical vote}
%\anote{mention: we use common random string}
In a voting scheme, we would like to allow users to vote for various candidates, while ensuring that their choices are private (\emph{privacy)}, and also ensuring that they cannot vote twice (\emph{uniqueness}).  Finally, we would also like to ensure that the final tallies are calculated correctly. The best possible notion for the latter is the so-called \emph{universal verifiability} setting, where voters post their votes on a public bulletin board, and anyone in the world is able to verify all the posted votes (and hence the final tallies). We note that universal verifiability cannot be achieved by voting schemes such as the ones based on homomorphic encryption, since the votes will need to be decrypted by the owner of the secret key, and revealing this secret key would trivially defeat privacy of votes. In fact, the combination of all three of these properties, (i) universal verifiability, (ii) privacy, (iii) uniqueness, cannot be achieved by any fully classical voting scheme. If the casting of a particular voting token is indistinguishable (to everyone) from casting of a fresh voting token, then an adversary can simply vote twice as follows. It creates copies of its initial voting token, rerandomizes these copies, and votes twice (or more) using all these tokens. The bulletin board cannot reject the second vote, since due to privacy, from its viewpoint, these votes are indistinguishable from someone else's (who has not voted yet) vote. Note that the impossibility persists even in trusted setup and ideal oracle models - it only (and crucially) relies on cloning the voting tokens (i.e. rewinding an adversary). However, this means that the attack does not work if the voting tokens are (unclonable) quantum states. This brings us to our construction. Our construction will be similar to our untraceable quantum money construction. A voting token will be the tuple of states $\ket{T_{ct, 1} \cdot \cansbsp} \otimes \dots \otimes \ket{T_{ct, 2\cdot \lambda} \cdot \cansbsp}$ (together with the serial number $ct$) where $(T_{ct,1}, \dots, T_{ct,2\cdot \lambda})$ is a tuple of full rank linear maps sampled using the randomness $F(K, ct).$ After rerandomizing these states (same as in our money scheme, generalized to $2\cdot \lambda$ states), to vote for a particular candidate $c \in [\lambda]$, the user samples a random tag\footnote{A serial number for the cast votes that the voter creates} $r \samp \zo^{\lambda}$ and measures the $i$-th state in the computational basis if the $i$-th bit $(c || r)_i = 0$, and in the Hadamard basis (i.e. QFT-and-measure) if $(c || r)_i = 1$. The measurement results $v_1, \dots, v_{2\cdot \lambda}$ along with the rerandomized serial number $ct'$ and $c, r$ will be the user's (classical) vote for the candidate $c$. Anyone can verify a vote by simply checking $\mathsf{PMem}(vk, c || r, v_1 || \dots || v_{2\cdot \lambda}) = 1$. Thus, our scheme satisfies universal verifiability.

Our scheme also satisfies privacy, against all parties, including the government/voting authority that creates the voting tokens (and it is allowed to create the tokens and keys maliciously). The privacy proof follows in the same way as the untraceablity proof of our quantum money scheme. Finally, uniqueness. For this property, we rely on the so-called \emph{direct product hardness} of subspace states \cite{BenDavid2023quantumtokens}, which says that no efficient adversary, given $\ket{A}$, can produce vectors $v, w$ such that $v \in A$ (computational basis) and $w \in A^\perp$ (Hadamard basis). Through the same arguments as in the unclonability proof of our untraceable money scheme (where reduced unclonability of our scheme to $1 \to 2$ unclonability of subspace states), we will be able to reduce the uniqueness security of our scheme to the direct product hardness of subspace states. In particular, we will show that no efficient adversary, given $k$ voting tokens (modeling a group of malicious voters cooperating) can output $k+1$ valid cast votes, each with different tags $r$. We note that this in particular (with $k = 0$) means that non-eligible parties (i.e. parties who did not receive a token) will not be able vote at all. Thus, during tallying/verification, all one needs to do is (i) verify each submitted cast vote $V$ using the public key, (ii) check that the exact same tag $r$ has not appeared before\footnote{If it has, only the first counts.}.

Finally, we re-emphasize that the cast votes in our scheme are classical, meaning that the voters will not need to transmit quantum states to vote. %tGoing further, while we do not formalize this, we note that the distribution of our voting tokens can also be done through (interactive) classical communication, using a \emph{remote state preparation} protocol \cite{chevalier2023semi}.

We refer the reader to \cref{sec:qvs} for the full construction and the proofs.

\subsection{Rerandomizable Encryption with Strong Correctness}\label{sec:techstrongrerand}
In this section, we introduce the notion of rerandomizable public-key encryption with \emph{strong rerandomization correctness} and discuss our solution. For rerandomizable encryption, the most common security notion is the following. In the security game, the challenger samples a random a bit $b \in \zo$, and if $b = 0$, the adversary (who has $pk$) is given an honest encryption $ct \samp \rpke.\mathsf{Enc}(pk, m)$ of a message $m$ of its choice, and also an honest rerandomization $ct' \samp \rpke.\mathsf{ReRand}(pk, ct)$ of $ct$. If $b = 1$, the adversary is given the honest encryption $ct \samp \rpke.\mathsf{Enc}(pk, m)$, and an independent fresh encryption $ct' \samp \rpke.\mathsf{Enc}(pk, m)$. We require that no adversary can predict the challenge case $b$ with probability better than $1/2 + \negl(\lambda).$ A strengthening of this notion requires the same for any string $ct$ (rather than for honestly randomly sampled ciphertexts).

\paragraph{\textbf{Definitions}} However, as discussed previously, this rerandomization guarantee alone (while sufficient for anonymity) is not sufficient for our unclonability proofs. On top of this guarantee, we also require that an efficient adversary cannot come up with a malicious ciphertext and a malicious random tape for $\rpke.\mathsf{ReRand}$ that leads to different decryptions before and after rerandomization. We say that such a scheme satisfies (I) \emph{strong (rerandomization) correctness}. While this notion is interesting in its own right, it will not be sufficient in $\io$-based constructions. While efficient adversaries cannot come up with \emph{bad} ciphertexts, note that such ciphertexts might still exist, whereas in $\io$-based security proofs one needs to  prove that two circuits (e.g., the initial circuit and the one that rejects bad ciphertexts) are completely equivalent - as opposed to just proving the hardness of finding an input where the circuits have different output. Thus, we also introduce the notion of (II) \emph{strong correctness with public testing}, where the public key includes a testing key. We require that there \emph{does not exist} a malicious ciphertext $ct$ and a malicious random tape $s$ such that
\begin{itemize}
    \item $\rpke.\mathsf{Test}(pk, ct) = \mathsf{GOOD}$.
    \item $\rpke.\mathsf{Dec}(sk, ct) \neq \rpke.\mathsf{Dec}(sk, \rpke.\mathsf{ReRand}(pk, ct); s)$
\end{itemize}
Since the testing key is public, an encryption scheme satisfying this notion can be utilized even with $\io$-based proofs, since the original obfuscated circuit in the actual construction can already test and reject bad ciphertexts. \begin{comment}
    A remark about our notation is in order. We will consider a separate \emph{protected} decryption algorithm $\mathsf{ProtectedDec}$ and a separate \emph{protected} decryption key $ssk$. The reason is that in our main application, unclonability of our money scheme, the tracing authority will have the decryption key of $\rpke$, while we also need to maintain strong correctness against it (since our unclonability game gives the tracing key to the adversary). Thus, we design our model so that there is a separate protected decryption mode that satisfies strong correctness even against parties that have the normal decryption key. The normal and the protected decryption key will have no difference for honest ciphertexts, and if one is not concerned about strong correctness against parties with $sk$, we can simply merge the two keys.
\end{comment}

We introduce two more stronger notions. First, we consider (III) \emph{strong correctness with simulatable (public) testing keys}. This notion is the same as our previous notion of public testing, with the addition that we split public key into two parts (encryption/rerandomization key and testing key), and we require that the testing key can be simulated (indistinguishable from real testing keys) using only the public encryption key (in fact, our constructions will need only the security parameter $1^\lambda$ to simulate). Such a notion is useful in our untraceable quantum money and voting schemes, since we will not be able to include the testing key in the truly uniformly random common string (CRS). We finally introduce (IV) \emph{strong correctness with simulatable all-accept testing keys.} In this case, we require that the scheme has a simulatable testing key such that it accepts \emph{all} strings $ct$ as good. This notion allows us to construct schemes (even with $\io$-based proofs) that do not use any ciphertext testing (or testing key) at all - at least in the original construction. Since the simulated keys always accept, the original constructions do not need to test, whereas in the proof (in an indistinguishable manner) we can start testing ciphertexts and reject bad ones, thus rely on the strong correctness guarantees discussed above.

We note that there are no existing post-quantum rerandomizable encryption schemes that satisfies any of these four security notions (even the simplest one, plain strong correctness (I)). We also note that approaches like fully homomorphic encryption or bootstrapping does not work, since they use honest randomness to \emph{refresh} cipertexts that have been honestly sampled and have been honestly computed on - no maliciousness is allowed in any part. Bad ciphertexts for such schemes can easily be produced efficiently.

In this work, we construct various schemes that satisfy all of the notions defined above, based solely on hardness of Learning with Errors \cite{regev2009lattices}.

\paragraph{\textbf{Background on Regev's PKE}} Our starting point is the well-known Regev's public-key encryption scheme \cite{regev2009lattices}, based on the Learning with Errors (LWE) problem. Let $n(\lambda), m(\lambda)$ denote the vector space dimension and the number of samples respectively, $\chi(\lambda)$ denote the error distribution, $B(\lambda)$ denote the bounds for the values in the support of $\chi(\lambda)$ (i.e. $\chi(\lambda) \subseteq [-B, B]$) and $q(\lambda)$ denote the modulus. The scheme is as follows, where all operations are performed modulo $q$ (i.e., in $\Z_q$). The public key is $(A, y)$ and the secret key is $s$, where $A \samp \Z_q^{n \times m}$, $s \samp \Z_q^n$, $e \samp \chi^m$ and $y^\mathsf{T} = s^\mathsf{T} \cdot A + e^\mathsf{T}$. To encrypt a bit $b \in \zo$, we sample a random \emph{binary} vector $r \samp \zo^{m}$ and output $(A\cdot r, y^\mathsf{T}\cdot r + b \cdot \floor{\frac{q}{2}})$. To rerandomize a ciphertext $ct = (\vec{a}, c)$, we simply add a fresh encryption of $0$ to it, which gives us $(\vec{a} + A\cdot r', c + y^\mathsf{T}\cdot r')$. Using LWE and leftover hash lemma \cite{impagliazzo1989pseudo}, one can show that rerandomization security is satisfied for \emph{any} $ct$. Finally, to decrypt $ct = (\vec{a}, c)$, we compute $|c - s^\mathsf{T}\cdot \vec{a}|$ and decrypt to $0$ if the value is less than $q/4$, and decrypt to $1$ otherwise.

\paragraph{\textbf{Thwarting Malicious Rerandomization Attacks with Hidden Shifts}} While Regev's PKE satisfies rerandomization security, it does not satisfy even the plain strong correctness property (I). Note that even honestly rerandomizing an honest ciphertext roughly $O(q/B)$ times is likely to produce a \emph{bad} ciphertext (decrypting to a different value after one more rerandomization). Thus, as our first step, we move to the \emph{subexponential modulus-to-noise ratio} LWE regime where $q = \subexp(\lambda)$, and $B = \poly(\lambda)$. Then, no polynomial number of honest rerandomizations on honest ciphertexts will produce bad ciphertexts. However, the adversary can still easily produce a malicious bad ciphertext: for example, consider $(\vec{a}, q/4 - 2 - k)$ where $\vec{a}$ is any vector and $k$ is a random value in $[- m\cdot B, m\cdot B]$. This ciphertext will decrypt to $0$ initially, whereas after rerandomizing it will decrypt to $1$ with probability roughly $1/k = 1/\poly(\lambda)$.

Our solution is to use \emph{hidden shifts}. We update our scheme so that the protected decryption key includes a hidden shift value $L$, selected uniformly at random from $[0, q/16].$ To decrypt $ct = (\vec{a}, c)$, we now test $|c - s^\mathsf{T}\cdot \vec{a} - \textcolor{red}{L}| < q/4$. Now, we prove that a ciphertext $(\vec{a}, c)$ can be bad only if
\begin{align*}
    &\text{Turns }0 \to 1: c - s^\mathsf{T}\cdot \vec{a} \in [q/4 - m \cdot B + L, q/4 - 1 + L]  \cup [3q/4 + 1 + L, 3q/4 + m \cdot B + L]\\
    &\text{Turns }1 \to 0: c - s^\mathsf{T}\cdot \vec{a} \in [q/4 + L, q/4 + m \cdot B - 1 + L] \cup [3q/4 - m\cdot B + 1 + L, 3q/4 + L]
\end{align*}
Since $L$ is unpredictable, no adversary (even given the secret $s$) will be able to produce a bad ciphertext except with subexponentially small probability.

\paragraph{\textbf{Testing Ciphertexts and Obfuscating Evasive Programs}} As discussed above, in most applications, we will also need to be able to publicly test if a given ciphertext is bad, that is, if it can decrypt to different values before and after (malicious) rerandomization. By above, (once $s$ and $\vec{a}$ are fixed) the set of bad ciphertexts $(\vec{a}, c)$ is actually \emph{evasive} - the program that checks for bad ciphertexts would output $0$ (meaning, \emph{not bad}) almost everywhere, except for a set of size $O(m\cdot B) = \poly(\lambda).$ This means that if we can obfuscate this program, we would be done. Unfortunately, obfuscating evasive programs in general is impossible \cite{barak2014obfuscation}. However, observe that in our case, the set of accepting points actually has a nice structure: they are consecutive! As a result, we actually only need point function obfuscation. We can simply use the point function obfuscation for the point $q/4 + L$. To test a ciphertext $(\vec{a}, c)$, we can simply run the point function with values $c - s^\mathsf{T}\cdot \vec{a} + m\cdot B, c - s^\mathsf{T}\cdot \vec{a} + m \cdot B + 1, \dots, c - s^\mathsf{T}\cdot \vec{a} - m \cdot B + 1$ for the lower two ranges, and then similarly with $c - s^\mathsf{T}\cdot \vec{a} + m\cdot B - 1 - q/2, \dots, c - s^\mathsf{T}\cdot \vec{a} - m \cdot B - q/2$ for the higher two ranges. Note that this is only $\leq 4 m \cdot B = \poly(\lambda)$ tests. The only remaining issue is that we are actually using the secret decryption key $s$ for testing. Our solution is that, instead of point function obfuscation, above we  use \emph{compute-and-compare (CC) obfuscation} \cite{WZ17} (which can be constructed from LWE) and obfuscate the compute-and-compare function $f_s(x) =^? \mathsf{TARGET}$ where $f_s((\vec{a}, c)) = c - s^\mathsf{T}a$ and the target point is $\mathsf{TARGET} = q/4 + L$. Since $q/4 + L$ is subexponentially unpredictable, by CC obfuscation guarantee, our public testing program will hide both $L$ and $s$!

\paragraph{\textbf{Pseudorandom Public Keys and Simulatable Testing Keys}} For our untraceable schemes, we need that the public keys of our scheme be pseudorandom. Our scheme described above can already be shown to have pseudorandom encryption keys $(A, y^\mathsf{T})$ by LWE. However, the ciphertext testing key is an obfuscated evasive program which outputs $0$ almost everywhere. Thus, it obviously is not a pseudorandom string. However, due to CC obfuscation security, we can actually simulate our testing key by obfuscating the program that always outputs $0$! This means that we can simply use the encryption key (which is pseudorandom) as our public encryption key, while users can simulate the ciphertext testing key on their own.

\section{Related Work}

Aaroson-Christiano \cite{AC12} constructs public-key quantum money from subspace states in the ideal classical oracle model. Zhandry \cite{Z19} later showed that the same construction is still secure when instantiated with indistinguishability obfuscation.

In the private-key setting; Ji et al \cite{JLS18} shows that pseudorandom quantum states can be constructed from one-way functions. Since pseudorandom quantum states are indistinguishable from Haar random states, which gives identical/pure state quantum money in the private-key setting, they automatically satisfy anonymity. Alagic et al \cite{amr20} constructs anonymous quantum money in private key setting with a stateful (updated each time a new banknote is issued) bank. Behera and Sattath \cite{bs21} gave a construction where banknotes are pseudorandom states and a received banknote is \emph{compared} to existing banknotes to verify.

Okamoto et al \cite{okamoto2008quantum} propose a quantum voting scheme (based conjugate coding) with quantum votes that satisfies privacy and uniqueness, but not universal verifiability (only the voting token creator can verify votes). There are also various classical voting schemes based on various assumptions \cite{adida2008helios, foo}, however, as discussed, they cannot satisfy all the desired properties at the same time.

\section{Preliminaries} 
\subsection{Notation}
Unless otherwise specified, adversaries are stateful quantum polynomial time (QPT) and our cryptographic assumptions are implicitly post-quantum. We write $\samp$ to denote a random sampling from some distribution or uniform sampling from a set.
 
\subsection{Digital Signature Schemes}
In this section we introduce the basic definitions of signatures schemes.

\begin{definition}\label{def:digdefn}
A digital signature scheme with message space $\mathcal{M}$ consists of the following algorithms that satisfy the correctness and security guarantees below.
\begin{itemize}
    \item $\mathsf{Setup}(1^\lambda):$ Outputs a signing key $sk$ and a verification key $vk$.
    \item $\mathsf{Sign}(sk, m):$ Takes the signing key $sk$, returns a signature for $m$.
    \item $\mathsf{Verify}(vk, m, s):$ Takes the public verification key $vk$, a message $m$ and supposed signature $s$ for $m$, outputs $1$ if $s$ is a valid signature for $m$.
\end{itemize}
\paragraph{Correctness}
We require the following for all messages $m \in \mathcal{M}$.
\begin{equation*}
    \Pr[\mathsf{Verify}(vk, m, s) = 1 : \begin{array}{c}
         sk, vk \samp \mathsf{Setup}(1^\lambda) \\
         s \samp \mathsf{Sign}(sk, m)
    \end{array}] = 1.
\end{equation*}
\paragraph{Adaptive existential-unforgability security under chosen message attack (EUF-CMA)}
Any QPT adversary $\adve$ with \emph{classical} access to the signing oracle has negligible advantage in the following game.
\begin{enumerate}
    \item Challenger samples the keys $sk, vk \samp \mathsf{Setup}(1)$.
    \item $\adve$ receives $vk$, interacts with the signing oracle by sending classical messages and receiving the corresponding signatures.
    \item $\adve$ outputs a message $m$ that it has not queried the oracle with and a forged signature $s$ for $m$.
    \item The challenger outputs $1$ if and only if $\Ver(vk, m, s) = 1$.
\end{enumerate}
If $\adve$ outputs the message $m$ before the challenger samples the keys, we call it \emph{selective EUF-CMA} security.
\end{definition}

\subsection{Puncturable Pseudorandom Functions}\label{sec:prf}
In this section, we recall puncturable pseudorandom functions.
\begin{definition}[\cite{SW14}]\label{defn:puncprf}
    A puncturable pseudorandom function (PRF) is a family of functions $\{F: \zo^{c(
\lambda)} \times \zo^{m(\lambda)} \to \zo^{n(\lambda)}\}_{\lambda \in \N^+}$ with the following efficient algorithms.
    \begin{itemize}
        \item $F.\mathsf{Setup}(1^\lambda):$ Takes in a security parameter and outputs a key in $\zo^{c(\lambda)}$.
        \item $F(K, x):$\footnote{We overload the notation and write $F$ to both denote  the function itself and the evaluation algorithm.} Takes in a key and an input, outputs an evaluation of the PRF.
        \item $F.\mathsf{Puncture}(K, S):$ Takes as input a key and a set $S \subseteq \zo^{m(\lambda)}$, outputs a punctured key.
    \end{itemize}
    
    We require the following.
    \paragraph{Correctness.} For all efficient distributions $\mathcal{D}(1^\lambda)$ over the power set $2^{\zo^{m(\lambda)}}$, we require
    \begin{equation*}
        \Pr[\forall x \not\in S~~F(K_S, x) = F(K, x): \begin{array}{c}
              S \samp \mathcal{D}(1^\lambda) \\
              K \samp \keygen(1^\lambda) \\
              K_S \samp \mathsf{Puncture}(K, S)
        \end{array}] = 1.
    \end{equation*}
    \paragraph{Puncturing Security} We require that any stateful QPT adversary $\adve$ wins the following game with probability at most $1/2 + \negl(\lambda)$.
    \begin{enumerate}
        \item $\adve$ outputs a set $S$.
        \item The challenger samples $K \samp \keygen(1^\lambda)$ and $ K_S \samp \mathsf{Puncture}(K, S)$
        \item The challenger samples $b \samp \zo$. If $b = 0$, the challenger submits $K_S, \{F(K, x)\}_{x \in S}$ to the adversary. Otherwise, it submits $K_S, \{y_s\}_{s \in S}$ to  the adversary where $y_s \samp \zo^{n(\lambda)}$ for all $s \in S$.
        \item The adversary outputs a guess $b'$ and we say that the adversary has won if $b' = b$.
    \end{enumerate}
\end{definition}

\begin{theorem}[\cite{SW14,GGM86,Z12}]\label{thm:puncprfexists}
Let $n(\cdot), m(\cdot)$ be efficiently computable functions.
\begin{itemize}
    \item If (post-quantum) one-way functions exist, then there exists a (post-quantum) puncturable PRF with input space $\zo^{m(\lambda)}$ and output space $\zo^{n(\lambda)}$.

\item If we assume subexponentially-secure (post-quantum) one-way functions exist, then for any $c > 0$, there exists a (post-quantum) $2^{-\lambda^c}$-secure\footnote{While the original results are for negligible security against polynomial time adversaries, it is easy to see that they carry over to subexponential security. Further, by scaling the security parameter by a polynomial and simple input/output conversions, subexponentially secure (for any exponent $c'$) one-way functions is sufficient to construct for any $c$ a puncturable PRF that is $2^{-\lambda^c}$-secure.} puncturable PRF against subexponential time adversaries, with PRF input space $\zo^{m(\lambda)}$ and output space $\zo^{n(\lambda)}$.
\end{itemize}
\end{theorem}

\subsection{Indistinguishability Obfuscation}
In this section, we recall indistinguishability obfuscation.
\begin{definition}
    An indistinguishability obfuscation scheme $\io$ for a class of circuits $\mathcal{C} = \{\mathcal{C}_\lambda\}_\lambda$ satisfies the following.
    \paragraph{Correctness.} For all $\lambda, C \in \mathcal{C}_\lambda$ and inputs $x$,
    $\Pr[\Tilde{C}(x) = C(x): \Tilde{C} \samp \io(1^\lambda, C)] = 1$.

\end{definition}
    \paragraph{Security.} Let $\mathcal{B}$ be any QPT algorithm that outputs two circuits $C_0, C_1 \in \mathcal{C}$ of the same size, along with auxiliary information, such that $\Pr[\forall x ~ C_0(x)=C_1(x) : (C_0, C_1, \regi{aux}) \samp \mathcal{B}(1^\lambda)] \geq 1 - \negl(\lambda)$. Then, for any QPT adversary $\mathcal{A}$,
    \begin{align*}
      \bigg|&\Pr[\adve(\io(1^\lambda, C_0), \regi{aux}) = 1 :  (C_0, C_1, \regi{aux}) \samp \mathcal{B}(1^\lambda)] -\\ &\Pr[\adve(\io(1^\lambda, C_1), \regi{aux}) = 1 : (C_0, C_1, \regi{aux}) \samp \mathcal{B}(1^\lambda)]\bigg| \leq \negl(\lambda).  
    \end{align*}

\subsection{Compute-and-Compare Obfuscation}
In this section, we recall the notion of compute-and-compare obfuscation \cite{WZ17}.
\begin{definition}[Compute-and-compare program]
     Let $f: \zo^{a(\lambda)} \to \zo^{b(\lambda)}$ be a function and $y \in \zo^{b(\lambda)}$ be a target value. The following program $P$, described by $(f, y, z)$, is called a \emph{compute-and-compare program.}

    \paragraph{$P(x):$}
    Compute $f(x)$ and compare it to $y$. If they are equal, output $1$. Otherwise, output $0$.
\end{definition}

A distribution $\mathcal{D}$ of such programs (along with quantum auxiliary information $\regi{aux}$) is called sub-exponentially unpredictable if for any QPT adversary, given the auxiliary information $\regi{aux}$ and the description of $f$, the adversary can predict the target value $y$ with at most subexponential probability.

\begin{definition}\label{defn:ccobf}
A compute-and-compare obfuscation scheme for a class of distributions  consists of efficient algorithms $\mathsf{CCObf.Obf}$ and $\mathsf{CCObf.Sim}$ that satisfy the following. Consider any distribution $\mathcal{D}$ over compute-and-compare programs, along with quantum auxiliary input, in this class.

\paragraph{Correctness.}
For any function $(f, y, z)$ in the support of $\mathcal{D}$,
$\Pr[\forall x ~ D'(x) = D(x) :\allowbreak D' \samp \mathsf{CCObf.Obf}(f, y)]\allowbreak \geq 1 - \negl(\lambda)$.

\paragraph{All-Zero Simulation} Program output by $\mathsf{CCObf.Sim}$ satisfy $P(x) = 0$ for all $x$.

\paragraph{Security}
$(\mathsf{CCObf.Obf}(f, y), \regi{aux}) \approx (\mathsf{CCObf.Sim}(1^\lambda, |f|, |y|), \regi{aux})$ where $(f, y), \regi{aux} \samp \mathcal{D}(1^\lambda)$.
\end{definition}

\begin{theorem}[\cite{WZ17,CLLZ21}]\label{thm:ccobf}
    Assuming the hardness of LWE, there exists compute-and-compare obfuscation for any class of sub-exponentially unpredictable distributions.
\end{theorem}

\subsection{Learning with Errors}\label{sec:lwe}
\begin{definition}
    Let $m(\lambda), n(\lambda), q(\lambda)$  be integers and $\chi(\lambda)$ be a probability distribution over $\Z_q$. The $(m, n, q, \chi)$-LWE assumption says that the following distributions are indistinguishable to any QPT adversary
    \begin{equation*}
(A, s^\mathsf{T}\cdot A+e) \approx_c (A, u)
    \end{equation*}
    where $A \samp \Z_q^{n \times m}, s \samp \Z_q^{n}$, $e \samp \chi^m$ and $u \samp \Z_q^{m}$.
\end{definition}

\subsection{Subspace States}
A subspace state is $\ket{A} = \sum_{v \in A}\ket{v}$ where $A$ is a subspace of the vector space $\F_2^n$. We will overload the notation and usually write $A, A^\perp$ to also denote the membership checking programs for the subspace $A$ and its orthogonal complement $A^\perp$.
\begin{theorem}[\protect{$1\to 2$-Unclonability \cite{Z19}}]\label{thm:onetotwo}
    Consider the following game between a challenger and an adversary $\adve$.

    \paragraph{\underline{$\mathsf{Exp}_\adve(1^\lambda)$}}
    \begin{enumerate}
    \item The challenger samples a subspace $A \leq \F_2^\lambda$ of dimension $\lambda/2$.
    \item The challenger submits $\ket{A}, \io(A), \io(A^\perp)$ to $\adve$.
    \item The adversary outputs a (entangled) bipartite register $\reg_1, \reg_2$.
    \item The challenger applies the projective measurement $\{\ketbra{A}{A}, I - \ketbra{A}{A}\}$, and outputs $1$ if the measurement succeeds.  Otherwise, it outputs $0$.
    \end{enumerate}
    Then, for any QPT adversary $\adve$, we have $\Pr[\mathsf{Exp}_\adve(1^\lambda) = 1] \leq \negl(\lambda)$.
\end{theorem}
\begin{theorem}[\protect{Direct Product Hardness \cite{sattath2022uncloneable}}]\label{thm:dphard}
 Consider the following game between a challenger and an adversary $\adve$.

    \paragraph{\underline{$\mathsf{Exp}_\adve(1^\lambda)$}}
    \begin{enumerate}
    \item The challenger samples a subspace $A \leq \F_2^\lambda$ of dimension $\lambda/2$.
    \item The challenger submits $\ket{A}, \io(A), \io(A^\perp)$ to $\adve$.
    \item The adversary outputs two vectors $v, w \in \F_2^\lambda$.
    \item The challenger checks if $v \in A$ and $w \in A^\perp$, and outputs $1$ if so. Otherwise, it outputs $0$.
    \end{enumerate}
    Then, for any QPT adversary $\adve$, we have $\Pr[\mathsf{Exp}_\adve(1^\lambda) = 1] \leq \negl(\lambda)$.
\end{theorem}

\section{Rerandomizable Encryption}\label{sec:rpke}
In this section, we first recall rerandomizable encryption, then introduce some new security notions and give secure constructions.

Recall that a rerandomizable encryption (RPKE) scheme is a public-key encryption with an additional efficient randomized algorithm $\mathsf{ReRandomize}(pk, ct)$. We require rerandomization security.

\begin{definition}[Rerandomization Security]\label{defn:rerandpke}\label{defn:trurandstatrerand}\label{lem:rerandfreshct}
A rerandomizable encryption scheme $\rpke$ is said to satisfy rerandomization security\footnote{Some work consider the weaker notion where an honest ciphertext is rerandomized rather than adversarial $s$.} if for any efficient adversary $\adve$, we have
\begin{equation*}
    \Pr[b' = b : \begin{array}{c}
         pk, sk \samp \rpke.\mathsf{Setup}(1^\lambda)  \\
         s \samp \adve(pk) \\
         s_0 \samp \rpke.\mathsf{ReRand}(pk, s) \\
         s_1 \samp \rpke.\mathsf{Enc}(pk, 0^{p(\lambda)}) \\
         b' \samp \adve(s_b)
    \end{array}] \leq \frac{1}{2} + \negl(\lambda).
\end{equation*}
\end{definition}
We also introduce a related security notion called \emph{statistical rerandomization security for truly random public keys}. For this, we require that if a truly random $pk$ is used instead of $\rpke.\mathsf{Setup}$, then the above security holds for unbounded $\adve$.

Now we introduce the notion of \emph{strong correctness}, which says no efficient adversary can find a malicious ciphertext and randomness tape such that decryption of the ciphertext before and after rerandomization differ.
\begin{definition}[Strong Correctness Security]\label{defn:strongcor}
A rerandomizable encryption scheme $\rpke$ is said to satisfy strong correctness if for any efficient adversary $\adve$, we have
\begin{equation*}
    \Pr[m \neq m' : \begin{array}{c}
         pk, sk \samp \rpke.\mathsf{Setup}(1^\lambda)  \\
         ct, r \samp \adve(pk) \\
         ct' = \rpke.\mathsf{ReRand}(pk, ct; r) \\
         m = \rpke.\mathsf{Dec}(sk, ct) \\
         m' = \rpke.\mathsf{Dec}(sk, ct')
    \end{array}] \leq \negl(\lambda).
\end{equation*}
\end{definition}

We now introduce a stronger notion, called \emph{strong correctness with public testing}. This requires that there exists a public ciphertext testing procedure such that there simply does not exist a malicious ciphertext (i) that passes the verification and (ii) whose decryption result before and rerandomization can differ. This can be considered a \emph{statistical} version, and it will be a useful property in $\io$-based proofs.
\begin{definition}[Strong Correctness with Public Testing]\label{defn:rerandpketest}
A rerandomizable encryption scheme $\rpke$ is said to satisfy strong correctness with public testing if
    \begin{equation*}
    \Pr_{pk, sk \samp \rpke.\mathsf{Setup}(1^\lambda)}[\exists ct, r \text{ s.t. } \begin{array}{c}
     \rpke.\mathsf{Test}(pk, ct) = 1~~\wedge  \\
         \rpke.\mathsf{Dec}(sk, ct) \neq  \rpke.\mathsf{Dec}(sk, \rpke.\mathsf{ReRand}(pk, ct; r))
    \end{array} ] \leq \negl(\lambda).
\end{equation*}
\end{definition}

Finally, we introduce two more notions that are useful in our constructions. The first is called \emph{simulatable testing key} property. In this setting, we consider the public testing key $ptk$ separately from the public encryption key $pk$. We require that there is a test key simulation algorithm such that $(\rpke.\mathsf{SimulateTestKey}(pk), pk)$ is (computationally) indistinguishable from $(ptk, pk)$ where $(pk, ptk, sk) \samp \rpke.\mathsf{Setup}(1^\lambda)$. Finally, if the simulated test key output by $\mathsf{SimulateTestKey}$ satisfies $\rpke.\mathsf{Test}(ptk, ct) = 1$ for all strings $ct$, then we call it \emph{simulatable all-accept testing key}. This property allows us to not use any ciphertext testing in actual constructions, and use $\rpke.\mathsf{Test}$ only in security proofs.

\subsection{Construction}\label{sec:pkecons}
In this section, we give a rerandomizable encryption scheme based on LWE (\cref{sec:lwe}) that satisfies all of the security properties discussed above:  (i) rerandomization security, (ii) strong correctness, (iii) strong correctness with public testing, (iv) simulatable all-accept testing keys, (v) pseudorandom public encryption keys and (vi) statistical rerandomization with truly random public keys. Also honest ciphertexts can be honestly rerandomized any polynomial number of times while staying \emph{good}.

Let $m(\lambda), n(\lambda)$ denote LWE dimension and number of samples parameters, $\chi(\lambda)$ denote the error distribution, $B(\lambda)$ denote the bounds for the values in the support of $\chi(\lambda)$ (i.e. $\chi(\lambda) \subseteq [-B, B]$) and $q(\lambda)$ denote the modulus; with $q(\lambda) = \subexp(\lambda)$ and $B = \poly(\lambda)$. All algebra in the scheme will implicitly be in $\Z_q$. Let $\ell(\lambda)$ denote the length of plaintexts. Finally, let $\mathsf{CCObf}$ be a compute-and-compare obfuscation scheme for subexponentially unpredictable distributions (\cref{defn:ccobf}), which exists assuming LWE (\cref{thm:ccobf}).

\paragraph{\underline{$\rpke.\mathsf{Setup}(1^\lambda)$}}
\begin{enumerate}
    \item Sample $A \samp \Z_q^{n \times m}$, $s \samp \Z_q^n$ and $e \samp \chi^m$.
    \item Compute $y^\mathsf{T} = s^\mathsf{T}\cdot A + e^\mathsf{T}$.
    \item For $i \in [\ell(\lambda)]$
    \begin{enumerate}[label=\arabic*.]
    \item Sample $L_i \samp [0, \floor{\frac{q}{16}}]$.
    \item Let $P_i$ be the compute-and-compare program that on input $(\vec{a}, c)$ computes $f_{s}(\vec{a}, c)$ and compares it to $\floor{\frac{q}{4}} + L_i$ and outputs $1$ on equality and $0$ otherwise; where $f_{s}(\vec{a}, c) \vcentcolon= (c - s^\mathsf{T}\cdot a)$.  Sample $\mathsf{OP}_i \samp \mathsf{CCObf}(P_i).$
    \end{enumerate}
    \item Set $pk = (A, y)$, $ptk = (\mathsf{OP}_i)_{i \in [\ell(\lambda)]}$ and $sk = (s, (L_i)_{i \in [\ell(\lambda)]})$. Output $pk, ptk, sk$.
    \end{enumerate}

    \paragraph{\underline{$\rpke.\mathsf{Enc}(pk, \mu)$}}
    \begin{enumerate}
    \item Parse $(A, y) = pk$.
        \item For $i \in [\ell(\lambda)]$
    \begin{enumerate}[label=\arabic*.]
    \item Sample $r_i \samp \zo^{m}$
    \item Compute $ct_i = (A\cdot r, y^\mathsf{T}\cdot r + \mu_i \cdot \floor{\frac{q}{2}})$.
    \end{enumerate}
    \item Output $(ct_i)_{i \in [\ell(\lambda)]}$.
    \end{enumerate}

    \paragraph{\underline{$\rpke.\mathsf{ReRand}(pk, ct)$}}
    \begin{enumerate}
        \item Sample $(ct'_i)_{i \in [\ell(\lambda)]} \samp \rpke.\mathsf{Enc}(pk, 0^{\ell(\lambda)})$.
        \item Output $(ct_i + ct'_i)_{i \in [\ell(\lambda)]}$.
    \end{enumerate}

    \paragraph{\underline{$\rpke.\mathsf{Test}(ptk, ct)$}}
    \begin{enumerate}
    \item Parse $(\mathsf{OP}_i)_{i \in [\ell(\lambda)]} = ptk$.
        \item Parse $((\vec{a}_i, c_i)_{i \in [\ell(\lambda)]} = ct$.
        \item For $i \in [\ell(\lambda)]$
        \begin{enumerate}[label=\arabic*.]
        \item For $sh \in [-m\cdot B + 1, m\cdot B] \bigcup [2\cdot\floor{\frac{q}{4}} -m\cdot B,  2\cdot\floor{\frac{q}{4}} -1 +m\cdot B] \subset \Z_q$, 
        check if $\mathsf{OP}_i((\vec{a}_i, c_i + sh)) = 1$, if not, output $0$ and terminate.
        \end{enumerate}
        \item Output $1$.
    \end{enumerate}

    \paragraph{\underline{$\rpke.\mathsf{SimulateTestKey}(1^\lambda)$}}
    \begin{enumerate}
        \item Output $\mathsf{CCObf.Sim}(1^\lambda, 1^{m(\lambda)}, 1^{n(\lambda)}, 1^{q(\lambda)})$
    \end{enumerate}

    \paragraph{\underline{$\rpke.\mathsf{Dec}(sk, ct)$}}
    \begin{enumerate}
    \item Parse $(s, (L_i)_{i \in [\ell(\lambda)]}) = sk$.
        \item Parse $((\vec{a}_i, c_i)_{i \in [\ell(\lambda)]} = ct$.
        \item For $i \in [\ell(\lambda)]$, check if $|c_i - s^\mathsf{T}\vec{a}_i - L_i| < \floor{\frac{q}{4}}$. If so, set $\mu_i = 0$, otherwise set $\mu_i = 1$.
        \item Output $\mu$.
    \end{enumerate}

\begin{theorem}
    $\rpke$ satisfies CPA security.
\end{theorem}
\begin{proof}
    Observe that the ciphertexts and the public key are the same as Regev's PKE, thus the result follows by \cite{regev2009lattices}.
\end{proof}
\begin{theorem}
    $\rpke$ satisfies rerandomization security, statistical rerandomization security for truly random public keys and pseudorandom public key property.
\end{theorem}
\begin{proof}
    The first prove the statistical rerandomization security. For simplicity we will consider single bit messages. Observe that a ciphertext $(\vec{a}, c)$ is rerandomized by adding $(A \cdot r, y^\mathsf{T}\cdot r)$ where $r \samp \zo^m$ and $(A, y)$ is the public key. However, when the public key (consider it as a matrix by adding $y^\mathsf{T}$ as a row to $A$) is truly random, then by a simple application of leftover hash lemma \cite{impagliazzo1989pseudo}, we get that $(A \cdot r, y^\mathsf{T}\cdot r)$ is also truly random. Similarly, a fresh encryption of $0$ will also be a truly (independent) random string by leftover hash lemma. Hence, the statistical reradomization follows.

    For pseudorandomness, observe that the public key is simply LWE samples. Thus by LWE security (\cref{sec:lwe}) it is indistinguishable from a truly random matrix.
    
    To see rerandomization security, we simply combine the above two properties.
\end{proof}

\begin{theorem}
    $\rpke$ satisfies strong rerandomization correctness, strong rerandomization correctness with testing and simulatable all-accept testing keys property.
\end{theorem}
\begin{proof}
    We will first prove strong rerandomization correctness with testing. For simplicity we will prove the single bit plaintext case $\ell = 1$, the general case follows by the same argument. 
    Let $(\vec{a}, c)$ be a (malicious) ciphertext. Since we decrypt it by checking $|c - s^\mathsf{T}\cdot \vec{a} - L| < \floor{\frac{q}{4}}$, observe that it decrypts to $0$ if and only if $c - s^\mathsf{T}\cdot \vec{a} \in [L, \floor{\frac{q}{4}} + L - 1] \bigcup [q - \floor{\frac{q}{4}} + 1 + L, q - 1 + L]$. Similarly, it decrypts to $1$ if and only if $c - s^\mathsf{T}\cdot \vec{a} \in [\floor{\frac{q}{4}} + L, q - \floor{\frac{q}{4}} + L]$. To rerandomize, we add $(A\cdot r, y^\mathsf{T}\cdot r)$ and get $(\vec{a} + A\cdot r, c + y^\mathsf{T}\cdot r)$. To decrypt this, we check $|c + y^\mathsf{T}\cdot r -  s^\mathsf{T}\cdot (\vec{a} + A\cdot r) - L| < \floor{\frac{q}{4}}$, but we have $c + y^\mathsf{T}\cdot r -  s^\mathsf{T}\cdot (\vec{a} + A\cdot r) - L = c - s^\mathsf{T}\cdot \vec{a} + e^\mathsf{T}\cdot r$ where $e$ is the error vector and $r \samp \zo^{m(\lambda)}$, since $y^\mathsf{T} = s^\mathsf{T}\cdot A + e^\mathsf{T}$. However, observe that $e^\mathsf{T}\cdot r \in [-m \cdot B, m \cdot B]$ since the error distribution satisfies $\chi \subseteq [-B, B]$. Thus, for decryption of the original ciphertext and the rerandomized ciphertext to potentially differ, the original value $c - s^\mathsf{T}\cdot \vec{a}$ needs to be near rounding thresholds. More formally, we need to have $c - s^\mathsf{T}\vec{a} \in [\floor{\frac{q}{4}} + L - m\cdot B, \floor{\frac{q}{4}} + L + m\cdot B - 1]$ or $c - s^\mathsf{T}\vec{a} \in [q - \floor{\frac{q}{4}} + L - m\cdot B + 1, q - \floor{\frac{q}{4}} + L + m\cdot B]$. By correctness of $\mathsf{CCObf}$, this is exactly what our test algorithm $\mathsf{Test}$ is checking, thus strong rerandomization correctness with testing follows.

    Now, observe that the target values $\floor{\frac{q}{4}} + L_i$ of the compute-and-compare programs $P_i$ are subexponentially unpredictable. Thus, by the security of $\mathsf{CCObf}$, the obfuscated programs $\mathsf{OP}_i$ are indistinguishable from obfuscations of all zero programs. Hence, simulatable all-accept testing property follows. Further, we can argue strong rerandomization correctness as follows. Consider the strong rerandomization correctness game and call it the first hybrid, $\hyb_0$. In $\hyb_1$, we modify the challenger so that the challenger first tests adversary's malicious ciphertext choice using $\rpke.\mathsf{Test}(ptk, s)$ where $ptk \samp \rpke.\mathsf{SimulateTestKey}(1^\lambda)$, and if it does not pass, the adversary automatically loses. However, by all-accepting simulated key property, $\hyb_0 \equiv \hyb_1$. In $\hyb_2$, we instead use the actual testing key rather than the simulated one. By simulated key security,  we have $\hyb_1 \approx \hyb_2$. Finally, we have $\Pr[\hyb_2 = 1] \leq 1/2 + \negl(\lambda)$ by  strong rerandomization correctness with testing security.
\end{proof}

\section{Definitions}\label{sec:defn}
In this section, we give the formal definitions for quantum money schemes that support privacy and tracing.
\begin{definition}[Quantum Money with Privacy and Tracing]
    A quantum money scheme $\bank$ consists of the following efficient algorithms.
    \begin{itemize}
        \item $\mathsf{Setup}(1^\lambda)$: Takes in a security parameter, outputs a minting key $mk$, a tracing key $tk$ and a public verification key $vk$.
        \item $\mathsf{GenBanknote}(mk, t)$: Takes in the minting key $mk$ and a tag $t$, outputs a quantum banknote register.
        \item $\mathsf{Verify}(vk,\reg)$: Takes in the public key and a quantum register $\reg$, outputs $0$ or $1$.
        \item $\mathsf{ReRandomize}(vk, \reg)$: Takes in the public key and a quantum register $\reg$, outputs the updated register.

        \item $\mathsf{Trace}(tk, \reg)$: Takes in the tracing key and a quantum register $\reg$, outputs a tag value or $\perp$.
    \end{itemize}
\end{definition}

We will require the following properties.

\begin{definition}[Correctness] For any tag $t \in \zo^{t(\lambda)}$, \begin{equation*}
        \Pr[b = 1 : \begin{array}{c}
             vk, mk, tk \samp \bank.\mathsf{Setup}(1^\lambda)  \\
             \reg \samp \bank.\mathsf{GenBanknote}(mk, t) \\
             b \samp \bank.\mathsf{Verify}(vk, \reg)
        \end{array}] \geq 1 - \negl(\lambda).
    \end{equation*}
\end{definition}

We will also consider \emph{projectiveness}, which requires that $\mathsf{Verify}(vk, (sn, \cdot))$ implements a rank-$1$ projector.

We now define correctness after rerandomization, which says that a valid banknote will stay valid after rerandomization.
    \begin{definition}[Correctness After Rerandomization] For any efficient algorithm $\mathcal{B}$, \begin{equation*}
        \Pr[b = 0 \vee b' = 1 : \begin{array}{c}
             vk, mk, tk \samp \bank.\mathsf{Setup}(1^\lambda)  \\
             \reg \samp \mathcal{B}(vk, mk, tk) \\
             b \samp \bank.\mathsf{Verify}(vk, \reg) \\
             \reg \samp \bank.\mathsf{ReRandomize}(vk, \reg) \\
             b' \samp \bank.\mathsf{Verify}(vk, \reg)
        \end{array}] \geq 1 - \negl(\lambda).
    \end{equation*}
\end{definition}

Finally, we will require that banknotes do not grow in size with rerandomization, and call this property \emph{compactness}\footnote{We note that, for compact schemes, we can also simply make re-randomization the last step of the verification algorithm}.

We now define counterfeiting (i.e unclonability) security for quantum money schemes. The security requires that any (QPT) adversary that obtains $k$ banknotes will note be able to produce $k + 1$ banknotes. While our security notion is quite similar to previous counterfeiting definitions (such as \cite{AC12}), we will require unclonability security even against an adversary that has the tracing key. We give the formal game-based definition in \cref{sec:extradef}.

\subsection{Fresh Banknote Security}
Previous work has defined a privacy notion, called \emph{anonymity}, where an adversary either gets back their banknotes in the original order or in permuted order (see \cref{defn:anonold}). We introduce a new, stronger security notion called \emph{indistinguishability from fresh banknotes}.

\begin{definition}[Indistinguishability from Fresh Banknotes]\label{defn:freshqm}
     Consider the following game between a challenger and an adversary $\adve$.

    \paragraph{\underline{$\qmindresh{\adve}(1^\lambda)$}}
    \begin{enumerate}
    \item Sample $vk, mk, tk \samp \bank.\mathsf{Setup}(1^\lambda)$.
    \item Submit $\textcolor{blue}{vk, mk}$ to $\adve$.
    \item Adversary $\adve$ outputs a register $\reg_0$.
    \item Run $\bank.\mathsf{Verify}(vk, \reg_0)$. If it fails, output $0$ and terminate.
    \item Run $\bank.\mathsf{ReRandomize}(vk, \reg_0)$.
    \item Sample $\reg_1 \samp \bank.\mathsf{GenBanknote}(mk, 0^{t(\lambda)})$.
    \item Submit $\reg_b$ to the adversary $\adve$.
    \item Adversary $\adve$ outputs a bit $b'$.
    \item Output $1$ if and only if $b' = b$.
    \end{enumerate}

    We say that the quantum money scheme $\bank$ satisfies \emph{fresh banknote indistinguishability} if for any QPT adversary $\adve$, we have 
    \begin{equation*}
        \Pr[\qmindresh{\adve}(1^\lambda) = 1] \leq \frac{1}{2} + \negl(\lambda).
    \end{equation*}
\end{definition}

\begin{theorem}\label{thm:freshbnimplanon}
    Any quantum money scheme that satisfies \emph{basic fresh banknote indistinguishability} also satisfies \emph{anonymity}.
\end{theorem}
The proof follows by a simple hybrid argument.

\subsection{Traceability}
In this section, we introduce tracing security for public-key quantum money.
\begin{definition}[Tracing Security]\label{defn:trace}
    Consider the following game between the challenger and an adversary $\adve$.
    \paragraph{$\underline{\tracegame{\adve}(1^\lambda)}$}
\begin{enumerate}
\item Initialize the list\footnote{In particular, it can contain the same element multiple times.} $\mathsf{TAGS} = []$.
    \item Sample $vk, mk, tk \samp \bank,\mathsf{Setup}(1^\lambda)$.
    \item Submit $\textcolor{blue}{vk, tk}$ to $\adve$.
    \item \underline{\textbf{Banknote Query Phase:}} For multiple rounds, $\adve$ queries for a banknote by sending a tag $t \in \zo^{t(\lambda)}$. For each query, the challenger executes $\regi{bn} \samp \bank.\mathsf{GenBanknote}(mk, t)$ and submits $\regi{bn}$ to the adversary. The challenger also adds $t$ to the list $\mathsf{TAGS}$.
    \item $\adve$ outputs a value $k$ and a $k$-partite register $(\reg_i)_{i \in [k]}$.
    \item For $i \in [k]$, the challenger tests $\bank.\mathsf{Verify}(vk, \reg_i) = 1$ and adds the output of $\bank.\mathsf{Dec}(tk, \reg_i)$ to the list $\mathsf{TAGS}'$. If any of the tests output $0$, the challenger outputs $0$ and terminates.
    \item The challenger checks if the list $\mathsf{SORT}(\mathsf{TAGS}')$ is a sublist of $\mathsf{SORT}(\mathsf{TAGS})$. If so, it outputs $0$. Otherwise, it outputs $1$.
\end{enumerate}
\end{definition}
\subsection{Untraceability}
In this section, we introduce the notion of \emph{untraceability} for quantum money for the first time. We will require that the banknotes are anonymous to everyone, including the malicious bank. We will consider this model in the common \emph{random} string model.
\begin{definition}[Untraceability]\label{defn:untrac}
     Consider the following game between a challenger and an adversary $\adve$.

    \paragraph{\underline{$\qmuntrac{\adve}(1^\lambda)$}}
    \begin{enumerate}
    \item Sample $crs \samp \zo^{q(\lambda)}$.
    \item Submit $crs$ to the adversary $\adve$.
    \item Adversary $\adve$ outputs keys $vk, mk$ and a register $\reg_0$.
    \item Run $\bank.\mathsf{Verify}(vk, \reg_0)$. If it fails, output $0$ and terminate.
    \item Sample $\reg_1 \samp \bank.\mathsf{GenBanknote}(mk)$.
    \item Run $\bank.\mathsf{Verify}(vk, \reg_1)$. If it fails, output $0$ and terminate.
    \item Sample $b \samp \zo$.
    \item Submit $\reg_b$ to the adversary $\adve$.
    \item Adversary $\adve$ outputs a bit $b'$.
    \item Output $1$ if and only if $b' = b$.
    \end{enumerate}

    We say that the quantum money scheme $\bank$ satisfies \emph{untraceability} if for any QPT adversary $\adve$, we have 
    \begin{equation*}
        \Pr[\qmuntrac{\adve}(1^\lambda) = 1] \leq \frac{1}{2} + \negl(\lambda).
    \end{equation*}
\end{definition}

\section{Construction with Anonymity and Traceability}
In this section, we give our public-key quantum money construction and prove that it satisfies unclonability, fresh banknote security, and tracing security.

We assume the existence of the following primitives that we use in our construction: 
(i) $\io$, subexponentially secure indistinguishability obfuscation,
(ii) $\rpke$, a rerandomizable public key encryption scheme with strong correctness and public testing (\cref{defn:rerandpketest})
    and (iii) $F$, a subexponentially secure puncturable PRF with input length $p_1(\lambda)$ and output length $p_2(\lambda)$.

We also utilize the following primitives that we only use in our security proofs: (i) $\ske$, a private-key encryption scheme with pseudorandom ciphertexts, and (ii) $\dss$, a signature scheme.

We also set the following parameters: $t(\lambda)$ to be the desired tag length, $sg(\lambda)$ to be the signature length of $\dss$ for messages of length $\lambda + t(\lambda)$, $c(\lambda)$ to be the ciphertext size of $\ske$ for messages of length $\lambda + sg(\lambda)$, $p_1(\lambda)$ to be the ciphertext size of $\rpke$ for messages of length $t(\lambda) + c(\lambda)$, $p_2(\lambda)$ to be the randomness size of the algorithm $\mathsf{SampleFullRank}$, and $p_3(\lambda)$ to be the randomness size of the algorithm $\pke.\mathsf{ReRandomize}$.

We now move onto our construction. Let $\cansbsp$ denote the \emph{canonical} $\lambda/2$-dimensional subspace of $\F_2^\lambda$ defined as $\mathsf{Span}(e_1, \dots, e_{\lambda/2})$ where $e_i \in \F_2^\lambda$ is the vector that has $1$ at the $i$-th index and $0$ at all the others.

\paragraph{\underline{$\bank.\mathsf{Setup}(1^\lambda)$}}
\begin{enumerate}
    \item Sample $K \samp F.\mathsf{Setup}(1^\lambda)$.
\item Sample $pk, sk \samp \rpke.\mathsf{Setup}(1^\lambda)$.
    \item Sample $\mathsf{OPMem} \samp \io(\mathsf{PMem})$ where $\mathsf{PMem}$ is the following program.
    
\begin{mdframed}
        {\bf $\underline{\mathsf{PMem}_{K}(id, v, b)}$}
        
        {\bf Hardcoded: $K$}
        \begin{enumerate}[label=\arabic*.]
            \item $T = \mathsf{SampleFullRank}(1^\lambda; F(K, id))$.
            \item Compute $w = T^{-1}(v)$ if $b = 0$; otherwise, compute $w = T^\mathsf{T}(v)$.
            \item Output $1$ if $w \in \cansbsp$ if $b = 0$ and if $w \in \cansbsp^\perp$ if $b = 1$. Otherwise, output $0$.
        \end{enumerate}
    \end{mdframed}

\item Sample $\mathsf{OPReRand} \samp \io(\mathsf{PReRand})$ where $\mathsf{PReRand}$ is the following program\footnote{We note that if we make the stronger assumption of $\rpke$ with \emph{all-accepting  simulatable testing keys}, we can actually remove the $\rpke.\mathsf{Test}$ line from the construction.}.
    
    \begin{mdframed}
        {\bf $\underline{\mathsf{PReRand}_{K}(id, s)}$}
        
        {\bf Hardcoded: $K, pk$}
        \begin{enumerate}[label=\arabic*.]
        \item Check if $\rpke.\mathsf{Test}(pk, id) = 1$. If not, output $\perp$ and terminate.
            \item $id' = \rpke.\mathsf{ReRand}(pk, id; s)$.
            \item $T_1 = \mathsf{SampleFullRank}(1^\lambda; F(K, id))$.
            \item $T_2 = \mathsf{SampleFullRank}(1^\lambda; F(K, id'))$.
            \item Output $id', T_2\cdot T_1^{-1}$.
        \end{enumerate}
\end{mdframed}

\item Set $vk = (\mathsf{OPMem}, \mathsf{OPReRand})$.
\item Set $mk = (K, pk)$.
\item Set $tk = sk$.
\item Output $vk, mk, tk$.
\end{enumerate}

\paragraph{\underline{$\bank.\mathsf{GenBanknote}(mk, tag)$}}
\begin{enumerate}
    \item Parse $(K, pk) = mk$.
    \item Sample $ict \samp \zo^{c(\lambda)}$.
    \item Sample $ct \samp \rpke.\mathsf{Enc}(pk, tag || ict)$.
    \item $T = \mathsf{SampleFullRank}(1^\lambda; F(K, ct))$.
    \item Set $\ket{\$} = \sum_{v \in A} \ket{T(v)}$.
    \item Output $ct, \ket{\$}$.
\end{enumerate}

\paragraph{\underline{$\bank.\mathsf{Verify}(vk, \reg)$}}
\begin{enumerate}
    \item Parse $(\mathsf{OPMem}, \mathsf{OPReRand}) = vk$.
    \item Parse $(id, \reg') = \reg$.
    \item Run $\mathsf{OPMem}$ coherently on $id, \reg', 0$. Check if the output is $1$, and then rewind (as in Gentle Measurement Lemma \cite{aarlemma}).
    \item Apply QFT to $\reg'$.
    \item Run $\mathsf{OPMem}$ coherently on $id, \reg', 1$. Check if the output is $1$, and then rewind.
    \item Output $1$ if both verifications passed above. Otherwise, output $0$.
\end{enumerate}

\paragraph{\underline{$\bank.\mathsf{ReRandomize}(vk, \reg)$}}
\begin{enumerate}
    \item Parse $(\mathsf{OPMem}, \mathsf{OPReRand}) = vk$.
    \item Parse $(ct, \reg') = \reg$.
    \item Sample $s \samp \zo^{p_3(\lambda)}$.
    \item $ct', T = \mathsf{OPReRand}(ct, s)$.
    \item Apply the linear map $T: \F_2^\lambda \to \F_2^\lambda$ coherently\footnote{That is, in superposition.  Note that since $T$ is an efficient (in both directions) bijection, we can indeed apply it in superposition with no garbage left.} to $\reg'$.
    \item Output $ct', \reg'$.
\end{enumerate}

\paragraph{\underline{$\bank.\mathsf{Trace}(tk, \reg)$}}
\begin{enumerate}
    \item Parse $sk = tk$.
    \item Parse $(ct, \reg') = \reg$.
    \item Compute $pl = \rpke.\mathsf{Dec}(sk, ct)$.
    \item If $pl = \bot$, output $\bot$. Otherwise, output first $t(\lambda)$ bits of $pl$.
\end{enumerate}

\begin{theorem}
    $\bank$ satisfies correctness after rerandomization and projectiveness.
\end{theorem}
See \cref{sec:projective} and \cref{sec:correctness} for the proofs.
\begin{theorem}
    $\bank$ satisfies fresh banknote indistinguishability (\cref{defn:freshqm}).
\end{theorem}
See \cref{defn:anonproof} for the proof.
\begin{theorem}
    $\bank$ satisfies tracing security (\cref{defn:trace}).
\end{theorem}
See \cref{sec:firsttracproof} for the proof.
\begin{theorem}
    $\bank$ satisfies counterfeiting security (\cref{defn:bankcf}).
\end{theorem}
See \cref{sec:proofunclon} for the proof.

\subsection{Projectiveness}\label{sec:projective}
Let $T$ be a full rank linear map. Observe that a vector $v$ satisfies $T^{-1}(v) \in \cansbsp$ if and only if $v \in A^*$ where $A^*$ is the set $\{T(w): w \in \cansbsp\}$, which is a subspace of dimension $\lambda/2$ since $T$ is a full rank linear map. Similarly, we can show that a vector $v$ satisfies $T^\mathsf{T}(v) \in \cansbsp^\perp$ if and only if $v \in (A^*)^\perp$. Note that $T^\mathsf{T}(v) \in \cansbsp^\perp$ if and only if $\langle T^\mathsf{T}\cdot v, u\rangle = 0$ for all $u \in \cansbsp$. However, this inner product is equal to $\langle v, T \cdot u\rangle$, and as $u$ ranges over all $\cansbsp$, $T\cdot u$ ranges over all $A^*$. Thus, we get that the above is equivalent to $\langle v, u'\rangle = 0$ for all $u' \in A^*$, which is equivalent to $v \in (A^*)^\perp$.

The above shows that our verification algorithm is equivalent to the subspace state verification algorithm (for the subspace $A^*$) of Aaronson-Christiano \cite{AC12}, which they prove implements a projection onto the subspace. Thus, projectiveness of our scheme follows.

\subsection{Correctness}\label{sec:correctness}
We prove correctness after rerandomization. For a banknote $ct, \ket{\psi}$ that has been verified, we know by projectiveness that $\ket{\psi} = \sum_{v \in \cansbsp} \ket{T(v)}$ where $T = \mathsf{SampleFullRank}(1^\lambda; F(K, ct))$. During rerandomization, $\mathsf{OPReRand}$ outputs $ct'$ and $T'' = T' \cdot T^{-1}$ where $T' = \sampfrm{F(K, ct')}$. Applying $T''$ in superposition, we get the state $\sum_{v \in \cansbsp} \ket{T''(v)}$, which is perfectly the state for the serial number $ct''$.

\subsection{Proof of Unclonability (Counterfeiting) Security}\label{sec:proofunclon}
We prove counterfeiting security (i.e. unclonability) through a sequence of hybrids, each of which is constructed by modifying the previous one. We implicitly pad all the $\io$ obfuscated programs to an appropriate size.

\paragraph{$\underline{\hyb_0}$}: The original game $\cfgame{\adve}(1^\lambda)$.

\paragraph{$\underline{\hyb_1}$}: First, at the beginning of the game, the challenger samples $isk \samp \ske.\mathsf{Setup}(1^\lambda)$.
 It also initializes a stateful counter $cnt = 0$ and a list $\mathsf{CT} = []$. Further, we modify the way the challenger mints banknotes. Instead of calling $\bank.\mathsf{GenBanknote}(mk, tag)$, it now executes the following subroutine for each query.
\paragraph{\underline{Minting Subroutine$(tag)$}}
\begin{enumerate}
 \item Parse $(K, pk) = mk$.
 \item \textcolor{red}{Add $1$ to $cnt$.}

 \item \textcolor{red}{ Sample $ict \samp \ske.\mathsf{Enc}(isk, cnt || 0^{sg(\lambda)})$.}

    \item Sample $ct \samp \rpke.\mathsf{Enc}(pk, tag || ict)$.
    \item \textcolor{red}{Add $ct$ to the list $\mathsf{CT}$.}
    \item $T = \mathsf{SampleFullRank}(1^\lambda; F(K, ct))$.
    \item Set $\ket{\$} = \sum_{v \in \cansbsp} \ket{T(v)}$.
    \item Output $ct, \ket{\$}$.
    \end{enumerate}
\paragraph{$\underline{\hyb_2}$}: We change the way the challenger verifies the banknotes output by the adversary. Instead of executing $\mathsf{Verify}$, it instead executes the following subroutine for each banknote.
\paragraph{\underline{Verify Subroutine$(\reg)$}}
\begin{enumerate}
        \item Parse $(\mathsf{OPMem}, \mathsf{OPReRand}) = vk$.
    \item Parse $(id, \reg') = \reg$.
    \textcolor{red}{\item Parse $pl_1 || pl_2 = \rpke.\mathsf{Dec}(sk, id)$ with $|pl_1| = t(\lambda), |pl_2| = c(\lambda)$.
    \item Compute $ipl_1 || ipl_2 = \ske.\mathsf{Dec}(isk, pl_2)$ with $|ipl_1| = \lambda$ and $|ipl_2| = sg(\lambda)$.
    \item Check if $ipl_1 \in [k]$. Output $0$ and terminate the subroutine if not.}
    \item Run $\mathsf{OPMem}$ coherently on $id, \reg', 0$. Check if the output is $1$, and then rewind.
    \item Apply QFT to $\reg'$.
    \item Run $\mathsf{OPMem}$ coherently on $id, \reg', 1$. Check if the output is $1$, and then rewind.
    \item Output $1$ if both verifications passed above. Otherwise, output $0$.
\end{enumerate}

\paragraph{$\underline{\hyb_3}$}: We now sample $i^* \samp [k]$ and also initialize the list $\mathsf{INDICES} = []$ at the beginning of the game. We also modify the verification subroutine so that each $ipl_1$ value is added to $\mathsf{INDICES}$. At the end of the game, the challenger (in addition to the previous checks) also checks if $i^*$ appears at least twice in $\mathsf{INDICES}$, and outputs $0$ if not.

\paragraph{$\underline{\hyb_{4}}$}: We sample $K' \samp F.\mathsf{Setup}(1^\lambda)$ and a random full rank linear map $T^{*}: \F_2^\lambda \to \F_2^\lambda$ at the beginning of the game. We also compute $A^* = T^{*}(\cansbsp)$ and create the following function/program.
\begin{equation*}
    M_{K', \mathsf{CT}_{i^*}}(ct) = \begin{cases}
        \mathsf{SampleFullRank}(1^\lambda; F(K', id)), \text{ if } id \neq \mathsf{CT}_{i^*}\\
        I, \text{ if } id = \mathsf{CT}_{i^*}\\
    \end{cases}
\end{equation*}
Further, we now sample $\mathsf{OPMem}$ and $\mathsf{OPReRand}$ as $\mathsf{OPMem} \samp \io(\mathsf{PMem}')$ and $\mathsf{OPReRand} \samp \io(\mathsf{PReRand}')$.

\begin{mdframed}
        {\bf $\underline{\mathsf{PMem}'_{K}(id, v, b)}$}
        
        {\bf Hardcoded: $K, \textcolor{red}{sk, isk, i^*, T^*, M_{K', \mathsf{CT}_{i^*}}}$}
        \begin{enumerate}[label=\arabic*.]
        \item \textcolor{red}{ Parse $pl_1 || pl_2 = \rpke.\mathsf{Dec}(sk, id)$ with $|pl_1| = t(\lambda), |pl_2| = c(\lambda)$.}
    \item \textcolor{red}{Parse $ipl_1 || ipl_2 = \ske.\mathsf{Dec}(isk, pl_2)$ with $|ipl_1| = \lambda, |ipl_2| = sg(\lambda)$.}
    \item \textcolor{red}{If $ipl_1 = i^*$, set $T =  M_{K', \mathsf{CT}_{i^*}}(id)\cdot T^*$.} Otherwise $T = \mathsf{SampleFullRank}(1^\lambda; F(K, id))$.
    
            \item Compute $w = T^{-1}(v)$ if $b = 0$; otherwise, compute $w = T^\mathsf{T}(v)$.
            \item Output $1$ if $w \in \cansbsp$ if $b = 0$ and if $w \in \cansbsp^\perp$ if $b = 1$. Otherwise, output $0$.
        \end{enumerate}
    \end{mdframed}

    \begin{mdframed}
        {\bf $\underline{\mathsf{PReRand}'_{K}(id, s)}$}
        
        {\bf Hardcoded: $K, pk, \textcolor{red}{sk, isk, i^*, T^*, M_{K', \mathsf{CT}_{i^*}}}$}
        \begin{enumerate}[label=\arabic*.]
          \item Check if $\rpke.\mathsf{Test}(pk, id) = 1$. Otherwise, output $\perp$ and terminate.
     \item \textcolor{red}{ Parse $pl_1 || pl_2 = \rpke.\mathsf{Dec}(sk, id)$ with $|pl_1| = t(\lambda), |pl_2| = c(\lambda)$.}
    \item \textcolor{red}{Parse $ipl_1 || ipl_2 = \ske.\mathsf{Dec}(isk, pl_2)$ with $|ipl_1| = \lambda, |ipl_2| = sg(\lambda)$.}
            \item $id' = \rpke.\mathsf{ReRand}(pk, id; s)$.
            \item \textcolor{red}{If $ipl_1 = i^*$, set $T_1 =  M_{K', \mathsf{CT}_{i^*}}(id)\cdot T^*$.} Otherwise, $T_1 = \mathsf{SampleFullRank}(1^\lambda; F(K, id))$.
            \item \textcolor{red}{ Parse $pl'_1 || pl'_2 = \rpke.\mathsf{Dec}(sk, id')$ with $|pl'_1| = t(\lambda), |pl'_2| = c(\lambda)$.}
    \item \textcolor{red}{Parse $ipl'_1 || ipl'_2 = \ske.\mathsf{Dec}(isk, pl'_2)$ with $|ipl'_1| = \lambda, |ipl'_2| = sg(\lambda)$.}
            \item \textcolor{red}{If $ipl_1' = i^*$, set $T_2 =  M_{K', \mathsf{CT}_{i^*}}(id')\cdot T^*$.} Otherwise, $T_2 = \mathsf{SampleFullRank}(1^\lambda; F(K, id'))$.
            \item Output $id', T_2\cdot T_1^{-1}$.
        \end{enumerate}
\end{mdframed}
Finally, we modify the minting subroutine as follows.
\paragraph{\underline{Minting Subroutine($\reg$)}}
\begin{enumerate}
 \item Parse $(K, pk) = mk$.
 \item {Add $1$ to $cnt$.}
 \item { Sample $ict \samp \ske.\mathsf{Enc}(isk, cnt || 0^{sg(\lambda)})$.}
 
    \item Sample $ct \samp \rpke.\mathsf{Enc}(pk, tag || ict)$.
    \item \textcolor{red}{If $cnt = i^*$, set $\ket{\$} = \sum_{v \in A^*} \ket{v}$ and jump to the final step.}
    \item $T = \mathsf{SampleFullRank}(1^\lambda; F(K, ct))$.
    \item $A = T(\cansbsp)$.
    \item Set $\ket{\$} = \sum_{v \in A} \ket{v}$.
    \item Output $ct, \ket{\$}$.
    \end{enumerate}

\paragraph{$\underline{\hyb_{5}}$}: At the beginning of the game, after we sample $T^*$, we also sample $\mathsf{P}_0 \samp \io(A^*)$ and $\mathsf{P}_1 \samp \io((A^*)^\perp)$. Further, we now sample $\mathsf{OPMem}$ as $\mathsf{OPMem} \samp \io(\mathsf{PMem}'')$ .

\begin{mdframed}
        {\bf $\underline{\mathsf{PMem}''_{K}(id, v, b)}$}
        
        {\bf Hardcoded: $K, sk, isk, i^*, \textcolor{red}{\mathsf{P}_0, \mathsf{P}_1}$}
        \begin{enumerate}[label=\arabic*.]
        \item Parse $pl_1 || pl_2 = \rpke.\mathsf{Dec}(sk, id)$ with $|pl_2| = c(\lambda)$.
    \item {Parse $ipl_1 || ipl_2 = \ske.\mathsf{Dec}(isk, pl_2)$ with $|ipl_1| = \lambda$ and $|ipl_3| = sg(\lambda)$. }
    \item \textcolor{red}{If $ipl_1 = i^*$,}
    \begin{enumerate}[label=\arabic*.]
        \item\textcolor{red}{Set $T = M_{K', \mathsf{CT}_{i^*}}(id)$.}
        \item\textcolor{red}{ Compute $w = T^{-1}(v)$ if $b = 0$; otherwise, compute $w = (T^{-1})^\mathsf{T}(v)$.}
    \item\textcolor{red}{Output the output $\mathsf{P}_b(v)$ and terminate.}
    \end{enumerate}
    
    \item If $ipl_1 \neq i^*$,
    \begin{enumerate}[label=\arabic*.]
        \item Set $T = \mathsf{SampleFullRank}(1^\lambda; F(K, id))$.
        \item Compute $w = T^{-1}(v)$ if $b = 0$; otherwise, compute $w = (T)^\mathsf{T}(v)$.
        \item Output $1$ if $w \in \cansbsp$ if $b = 0$ and if $w \in \cansbsp^\perp$ if $b = 1$. Otherwise, output $0$.
    \end{enumerate}
        \end{enumerate}
    \end{mdframed}

\paragraph{$\underline{\hyb_{6}}$}: We now sample $\mathsf{OPReRand}$ as $\mathsf{OPReRand} \samp \io(\mathsf{PReRand}'')$ .
    \begin{mdframed}
        {\bf $\underline{\mathsf{PReRand}''_{K}(id, s)}$}
        
        {\bf Hardcoded: $K, pk, {sk, isk, i^*}, M_{K', \mathsf{CT}_{i^*}}$}
        \begin{enumerate}[label=\arabic*.]
                \item Check if $\rpke.\mathsf{Test}(pk, id) = 1$. Otherwise, output $\perp$ and terminate.
     \item { Parse $pl_1 || pl_2 = \rpke.\mathsf{Dec}(sk, id)$ with $|pl_2| = c(\lambda)$.}
    \item {Parse $ipl_1 || ipl_2 = \ske.\mathsf{Dec}(isk, pl_2)$ with $|ipl_1| = \lambda$ and $|ipl_2| = sg(\lambda)$. }
            \item $id' = \rpke.\mathsf{ReRand}(pk, id; s)$.
            \item \textcolor{red}{If $ipl_1 = i^*$, set $T_1 =  M_{K', \mathsf{CT}_{i^*}}(id)$.} Otherwise, $T_1 = \mathsf{SampleFullRank}(1^\lambda; F(K, id))$.
            \item \textcolor{red}{If $ipl_1 = i^*$, set $T_2 =  M_{K', \mathsf{CT}_{i^*}}(id')$.} Otherwise, $T_2 = \mathsf{SampleFullRank}(1^\lambda; F(K, id'))$.
            \item Output $id', T_2\cdot T_1^{-1}$.
        \end{enumerate}
\end{mdframed}

\begin{lemma}
    $\hyb_0 \approx \hyb_1$.
\end{lemma}
\begin{proof}
    These hybrids differ in two places. First, the challenger executes the minting procedure directly instead of calling $\bank.\mathsf{GenBanknote}$. This is only a semantic change and makes no difference. Second, we replace the random strings $ict \samp \zo^{c(\lambda)}$ with ciphertexts of $\ske$. Thus, the security follows by the pseudorandom ciphertext security of $\ske$, since the experiments do not use $sk$ and can be simulated only using the ciphertexts.
\end{proof}

\begin{lemma}\label{lem:zerocopunlearn}
    $\hyb_1 \approx \hyb_2$.
\end{lemma}
\begin{proof}
We will show that by strong rerandomization correctness property, any banknote whose serial number decrypts to a value outside $[k]$ is \emph{rooted} (i.e is a rerandomization of) at a subspace state that was not even given to the adversary. Thus, the result then will follow by unlearnability of subspaces (we can also think of this as $0 \to 1$ unclonability).

We give a formal proof in \cref{sec:zerocopunlearn}.
\end{proof}

\begin{lemma}
    $\Pr[\hyb_3 = 1] \geq \frac{\Pr[\hyb_2 = 1]}{k}$
\end{lemma}
\begin{proof}
    Observe that all $ipl_1$ added to $\mathsf{INDICES}$ are required to be in $[k]$, whereas the list at the end will have size $k + 1$ (assuming the challenger has not terminated with output $0$ already). Thus, by pigeonhole principle, there is a value $i^{**} \in [k]$ such that  it appears in $\mathsf{INDICES}$ more than once. Our random guess $i^*$ will satisfy $i^* = i^{**}$ with probability $1/k$.
\end{proof}

\begin{lemma}\label{lem:switchtokprim}
    $\hyb_3 \approx \hyb_4$.
\end{lemma}
\begin{proof}
This follows through a hybrid argument using the puncturing security of the PRF and the security $\io$, where we create hybrids over all strings $id$. Through sufficient padding, the result follows by subexponential security of the PRF scheme and $\io$.

We give a formal proof in \cref{sec:switchtokprim}.
\end{proof}

\begin{lemma}
    $\hyb_4 \approx \hyb_5$.
\end{lemma}
\begin{proof}
We will show that the programs $\mathsf{PMem}'$ and $\mathsf{PMem}''$ have the same functionality. Then, the result follows by the security of $\io$. 

The behaviour of the two programs can possibly differ only on inputs such that $ipl = i^*$. However, by the same argument as in the proof of projectiveness of our scheme, we know that $\mathsf{PMem}'$ implements membership checking programs for the subspaces $A^*, (A^*)^\perp$ when $ipl = i^*$. By correctness of the obfuscation used to create $\mathsf{P}_0, \mathsf{P}_1$; $\mathsf{PMem}''$ does the same thing. Thus, we get that the programs $\mathsf{PMem}'$ and $\mathsf{PMem}''$ have the same functionality, and the result follows by the security of $\io$.
\end{proof}

\begin{lemma}\label{lem:unclonstrongrerand}
    $\hyb_5 \approx \hyb_6$.
\end{lemma}
\begin{proof}
Observe that by the strong rerandomization correctness of $\rpke$, there does not exist $id$ and $s$ such that $\rpke.\mathsf{Test}(pk, id) = 1$, but $id$ and $id' = \mathsf{ReRand}(pk, ct; s)$ decrypt to different values. Thus, in $\mathsf{PReRand}'$, we know that $id$ and $id'$ will decrypt to the same value $pl_2$, which will decrypt to the same value $ipl_1$ due to deterministic decryption of $\ske$.
Thus, separately decrypting $id'$ to obtain $ipl_1'$ versus directly using $ipl_1$ instead makes no difference. Further, this means that for both $id$ and $id'$, we will be in the same case with respect to the test $ipl_1 =^? i^*$. Therefore, the factor $T^*$ cancels outs when we compute $T_2 \cdot T_1^{-1} = (M_{K', \mathsf{CT}_{i^*}}(id')\cdot T^*) \cdot (M_{K', \mathsf{CT}_{i^*}}(id)\cdot T^*)^{-1}$. Therefore, the programs $\mathsf{PReRand}',\mathsf{PReRand}''$ have exactly the same functionality. The result follows by the security of $\io$.
\end{proof}

\begin{lemma}\label{lem:unclonlasth}
    $\Pr[\hyb_6 = 1] = \negl(\lambda)$.
\end{lemma}
\begin{proof}
Suppose otherwise for a contradiction.
This result follows by the $1 \to 2$ unclonability of the subspace state $\ket{A^*}$. Observe that this hybrid can be simulated using only a single copy of $\ket{A^*}$ along with $\io(A^*), \io((A^*)^\perp)$ (by sampling the other keys e.g. $K, K'$ ourselves). At the end, we know that at least two of the forged banknotes are such that their serial numbers decrypt to $i^*$. Let $ct_j$ and $ct_\ell$ denote the serial numbers of these banknotes. We also know that by projectiveness of the verification, these two banknote states (when they pass the verification) will be exactly $\ket{M_{K', \mathsf{CT}_{i^*}}(ct_j)\cdot A^*}$ and $\ket{M_{K', \mathsf{CT}_{i^*}}(ct_\ell)\cdot A^*}$. However, since we (the reduction) sample $K'$, we can actually take both of these back to $\ket{A^*}$ by applying $(M_{K', \mathsf{CT}_{i^*}}(ct_j))^{-1}, (M_{K', \mathsf{CT}_{i^*}}(ct_\ell))^{-1}$ in superposition, and obtain two copies of $\ket{A^*}$. However, this means that we cloned $\ket{A^*}$ with non-negligible probability, which is a contradiction by \cref{thm:onetotwo}.
\end{proof}

Now, suppose for a contradiction that there exists a QPT $\adve$ such that $\Pr[\cfgame{\adve}(1^\lambda) = 1]$ is non-negligible. Then, we get that $\Pr[\hyb_6 = 1]$ is also non-negligible, which is a contradiction by \cref{lem:unclonlasth}.

\subsection{Proof of Fresh Banknote Indistinguishability}\label{defn:anonproof}
We will prove security through a sequence of hybrids, each of which is obtained by modifying the previous one. Let $\adve$ be a QPT adversary.
\paragraph{$\underline{\hyb_0}$}: The original game $\qmindresh{\adve}(1^\lambda)$.

\paragraph{$\underline{\hyb_1}$}: In the re-randomization step (in Step 5), instead of performing $\mathsf{ReRandomize}(vk, \reg_0)$, we instead execute the following subroutine.
\paragraph{\underline{Rerandomize Subroutine}}
\begin{enumerate}
    \item Parse $(ct^*, \reg'_0) = \reg_0$.
    \item Sample $s^* \samp \zo^{p_3(\lambda)}$.
                \item $ct^{**}_0 = \rpke.\mathsf{ReRand}(pk, ct^*; s^*)$.
            \item $T_1 = \mathsf{SampleFullRank}(1^\lambda; F(K, ct^*))$.
            \item $T^*_0 = \mathsf{SampleFullRank}(1^\lambda; F(K, ct^{**}_0))$.
            \item Set $M^* = (T^*_0)\cdot T_1^{-1}$.
    \item Apply the linear map $M^*: \F_2^\lambda \to \F_2^\lambda$ coherently to $\reg'_0$.
    \item Set $\reg_0 = (ct^{**}_0, \reg'_0)$.
\end{enumerate}

\paragraph{$\underline{\hyb_2}$}: In the rerandomization subroutine, we replace the line
$$ct^{**}_0 = \rpke.\mathsf{ReRand}(pk, ct^*; s^*)$$
with
$$ct^{**}_0 \samp \rpke.\mathsf{Enc}(pk, 0^{t(\lambda) + c(\lambda)}).$$

\paragraph{$\underline{\hyb_3}$}: We further modify the rerandomization subroutine as follows.
\paragraph{\underline{Rerandomize Subroutine}}
\begin{enumerate}
 \item Parse $(ct^*, \reg'_0) = \reg_0$.
    \item Sample $ct^{**}_0 \samp \rpke.\mathsf{Enc}(pk, 0^{t(\lambda) + c(\lambda)})$.
            \item $T^*_0 = \mathsf{SampleFullRank}(1^\lambda; F(K, ct^{**}_0))$.
            \item $A_0^* = T^*_0(\cansbsp)$.
    \item Set $\reg'_0 = \sum_{v \in A^*}\ket{v}$.
    \item Set $\reg_0 = (ct^{**}_0, \reg'_0)$.
\end{enumerate}

\paragraph{$\underline{\hyb_4}$}: In the challenge phase, for the case $b = 1$, instead of executing
    $$\reg_1 \samp \mathsf{GenBanknote}(mk, 0^{t(\lambda)})$$
    we execute
    \begin{align*}
    &ct^{**}_1 \samp \rpke.\mathsf{Enc}(pk, 0^{t(\lambda) + c(\lambda)}) \ \\
    &T^*_1 = \mathsf{SampleFullRank}(1^\lambda; F(K, ct^{**}_1))\\
    &\reg'_1 =     \sum_{v \in \cansbsp}\ket{T^*_1(v)} \\
    &\reg_1 = (ct^{**}_1, \reg'_1)
    \end{align*}

\begin{lemma}
    $\hyb_0 \approx \hyb_1.$
\end{lemma}
\begin{proof}
    $\hyb_1$ simply unwraps $\mathsf{ReRandomize}(vk, \reg_0)$, including having the challenger execute the code of the program $\mathsf{PReRand}$ rather than using $\mathsf{OPReRand}$. By correctness of $\io$, the result follows.
\end{proof}

\begin{lemma}
$\hyb_1 \approx \hyb_2.$    
\end{lemma}
\begin{proof}
    The result follows by rerandomization security of $\rpke$ (\cref{lem:rerandfreshct}).
\end{proof}

\begin{lemma}
   $\hyb_2 \approx \hyb_3.$    
\end{lemma}
\begin{proof}
     As proven in the proof of correctness after rerandomization, once we rerandomize a banknote to the new serial number $ct_0^{**}$, our banknote state becomes the perfect banknote state associated with $ct^{**}$. Thus the result follows.
\end{proof}
\begin{lemma}
   $\hyb_3 \equiv \hyb_4.$    
\end{lemma}
\begin{proof}
     In this hybrid the challenger simply executes the minting itself, rather than calling $\bank.\mathsf{GenBanknote}$. This is only a semantic change.
\end{proof}
\begin{lemma}
    $\Pr[\hyb_4 = 1] \leq 1/2.$
\end{lemma}
\begin{proof}
    The results follows due to the fact that two challenge cases $b= 0$ and $b=1$ are completely symmetrical.
\end{proof}

Thus, we get $\Pr[\qmindresh{\adve}(1^\lambda) = 1] \leq 1/2 + \negl(\lambda)$, completing the proof.

\subsection{Proof of Tracing Security}\label{sec:firsttracproof}
We will prove security through a sequence of hybrids, each of which is constructed by modifying the previous one. Let $\adve$ be a QPT adversary.
\paragraph{$\underline{\hyb_0}$:} The original game $\tracegame{\adve}(1^\lambda)$.

\paragraph{$\underline{\hyb_1}$:} At the beginning of the game, the challenger initializes a stateful counter for tags, $\mathsf{cnt}(t)$, that is initially set to $0$ for each tag. It also initializes the lists $\mathsf{PL} = []$, $\mathsf{PL}' = []$. It also samples
\begin{align*}
    &isk \samp \ske.\mathsf{Setup}(1^\lambda)\\
    &ivk, isgk \samp \dss.\mathsf{Setup}(1^\lambda).
\end{align*}

We also modify the way the challenger answers banknote queries. Instead of executing $$\regi{bn} \samp \mathsf{GenBanknote}(mk, tag)$$ on a query for $tag$, it instead executes the following subroutine.
\paragraph{\underline{Minting Subroutine$(tag)$}}
\begin{enumerate}
    \item Parse $(K, pk) = mk$.
    \item \textcolor{red}{Increase $\mathsf{cnt}(tag)$ by $1$.}
    \item \textcolor{red}{Sample $sig \samp \dss.\mathsf{Sign}(sk, tag || \mathsf{cnt}(tag))$.}
    \item \textcolor{red}{Sample $ict \samp \ske.\mathsf{Enc}(isk, \mathsf{cnt}(tag) || sig)$.}
    \item \textcolor{red}{Add $tag || \mathsf{cnt}(tag)$ to the list $\mathsf{PL}$.}
    \item Sample $ct \samp \rpke.\mathsf{Enc}(pk, tag || ict)$.
    \item $T = \mathsf{SampleFullRank}(1^\lambda; F(K, ct))$.
    \item Set $\ket{\$} = \sum_{v \in \cansbsp} \ket{T(v)}$.
    \item Output $ct, \ket{\$}$.
\end{enumerate}

\paragraph{$\underline{\hyb_2}$:} We change the way the challenger verifies the banknotes output by the adversary. Instead of executing $\mathsf{Verify}$, it instead executes the following subroutine for each banknote.
\paragraph{\underline{Verify Subroutine$(\reg)$}}
\begin{enumerate}
        \item Parse $(\mathsf{OPMem}, \mathsf{OPReRand}) = vk$.
    \item Parse $(id, \reg') = \reg$.
    \textcolor{red}{\item Parse $pl_1 || pl_2 = \rpke.\mathsf{Dec}(sk, id)$ with $|pl_1| = t(\lambda), |pl_2| = c(\lambda)$.
    \item Parse $ipl_1 || ipl_2 = \ske.\mathsf{Dec}(isk, pl_2)$ with $|ipl_1| = \lambda$, $|ipl_2| = sg(\lambda)$.
    \item Check if $\dss.\mathsf{Verify}(ivk, pl_1 || ipl_1, ipl_2) = 1$. Output $0$ and terminate the subroutine if the checks fail.
    \item Add $ipl_1 || ipl_2$ to the list $\mathsf{PL}'$.
    }
    \item Run $\mathsf{OPMem}$ coherently on $id, \reg', 0$. Check if the output is $1$, and then rewind.
    \item Apply QFT to $\reg'$.
    \item Run $\mathsf{OPMem}$ coherently on $id, \reg', 1$. Check if the output is $1$, and then rewind.
    \item Output $1$ if both verifications passed above. Otherwise, output $0$.
\end{enumerate}

\paragraph{$\underline{\hyb_3}$:} In the verification subroutine, after checking $\dss.\mathsf{Verify}(ivk, pl_1 || ipl_1, ipl_2) = 1$, the challenger also checks if $ipl_1 || ipl_2 \in \mathsf{PL}$, and outputs $0$ and terminates if not.
 
\paragraph{$\underline{\hyb_4}$:} We add an additional winning condition at the end of the game: If the list $\mathsf{PL}'$ contains any duplicate elements, the challenger outputs $0$ and hence the adversary loses.

\begin{lemma}
    $\hyb_0 \approx \hyb_1.$
\end{lemma}
\begin{proof}
    Since the adversary does not have the secret key $isk$, the result follows by the semantic security of $\ske$.
\end{proof}

\begin{lemma}
    $\hyb_1 \approx \hyb_2.$
\end{lemma}
\begin{proof}
    As in \cref{lem:zerocopunlearn}, this follows due to $0 \to 1$ unclonability (i.e. unlearnability) of the subspace states. The proof follows similarly to \cref{lem:zerocopunlearn} (proven in \cref{sec:zerocopunlearn}). The main difference in this case is that we will need to perform the \emph{rooted at original}  banknote (banknotes whose serial numbers decrypt to one of the original tags) versus \emph{outside banknote} test for tag values chosen by the adversary, whereas in \cref{lem:zerocopunlearn} the original set was just $[k]$, which is efficiently. In this case, thanks to the hidden signatures $isig$, we can still verify that a tag is one of the original tags efficiently without having to hardcode all the tag values inside the $\io$ program.
\end{proof}

\begin{lemma}
    $\hyb_2 \approx \hyb_3.$
\end{lemma}
\begin{lemma}
    Observe that the adversary only obtains signatures for the elements of $\mathsf{PL}$. Thus, the result follows by the unforgeability security of the signature scheme $\dss$.
\end{lemma}

\begin{lemma}\label{lem:traceh2h3}
    $\hyb_3 \approx \hyb_4.$
\end{lemma}
\begin{proof}
    Observe that this is essentially unclonability of each indiviual banknote: Any duplicate value in $\mathsf{PL}'$ means that a banknote with that payload value was cloned. Thus, the proof follows similarly to the proof of unclonability of our scheme (especially hybrids $\hyb_3$ through $\hyb_6$ in \cref{sec:proofunclon}).
\end{proof}

\begin{lemma}
    $\Pr[\hyb_4 = 1] = 0.$
\end{lemma}
\begin{proof}
Observe that at the end of the game, if the challenger already has not terminated with $0$, we will have that $\mathsf{PL}'$ will be a subset of $\mathsf{PL}$. Further, for the adversary to win the game, each element in $\mathsf{PL}'$ can appear once. Finally, note that $\mathsf{PL}, \mathsf{PL}'$ actually contains the tags (with a multiplicity counter attached to each tag) of the banknotes given to the adversary and the banknotes output by the adversary, respectively. Thus, combining these means that $\mathsf{TAGS}'$ will contain the same tags as $\mathsf{TAGS}$, but fewer or equal multiplicity for each tag. This implies that $\mathsf{SORT}(\mathsf{TAGS}')$ is a sublist of $\mathsf{SORT}(\mathsf{TAGS})$. Hence, the adversary loses surely in $\hyb_4$.
\end{proof}

By above, we conclude that $\Pr[\tracegame{\adve}(1^\lambda) = 1] = \negl(\lambda)$ for any QPT adversary $\adve$.

\section{Construction with Untraceability}\label{sec:consuntrac}
In this section, we give our public-key quantum money construction with untraceability, in the common (uniformly) random string model. In our scheme, verification automatically rerandomizes a banknote.

We assume the existence of the following primitives that we use in our construction.
(i) $\io$, subexponentially secure indistinguishability obfuscation, (ii) $\rpke$, a rerandomizable public key encryption scheme with strong correctness, public testing, pseudorandom public encryption keys, simulatable testing keys and statistical rerandomization security for truly random public keys (\cref{defn:trurandstatrerand}), (iii) $F$, a subexponentially secure puncturable PRF with input length $p_1(\lambda)$ and output length $p_2(\lambda)$, and (iv) $\mathsf{NIZK}$, a non-interactive zero knowledge argument system in the common random string model (CRS) for the language $L$ defined below.

We also set the following parameters:  $p_1(\lambda)$ to be the ciphertext size of $\rpke$ for messages of length $\lambda$, $p_2(\lambda)$ to be the randomness size of the algorithm $\mathsf{SampleFullRank}$, $p_3(\lambda)$ to be the randomness size of the algorithm $\pke.\mathsf{ReRandomize}$, $q_1(\lambda)$ to be the CRS length of $\mathsf{NIZK}$, $q_2(\lambda)$ to be the public key size of $\rpke$, $q(\lambda) = q_1(\lambda) + q_2(\lambda)$ to be the CRS length of our scheme.

Finally, we define the $NP$ language $L$ as consisting of $x$ such that
\begin{equation*}
    \exists \textcolor{red}{K, r}~~ x = \io(\mathsf{PMem}_{\textcolor{red}{K}}; \textcolor{red}{r})
\end{equation*}

\paragraph{\underline{$\mathsf{Setup}(crs)$}}
\begin{enumerate}
\item Parse\footnote{See \cref{appn:remark} for a remark on representing the public key as a binary string.} $crs_{nizk} || pk = crs$ with $|crs_{nizk}| = q_1(\lambda)$ and $|pk| = q_2(\lambda).$
    \item Sample $K \samp F.\mathsf{Setup}(1^\lambda)$.
    \item Sample random tape $r_{iO}$ and compute $\mathsf{OPMem} = \io(\mathsf{PMem}; r_{iO})$ where $\mathsf{PMem}$ is the following program.
    
\begin{mdframed}
        {\bf $\underline{\mathsf{PMem}_{K}(id, v, b)}$}
        
        {\bf Hardcoded: $K$}
        \begin{enumerate}[label=\arabic*.]
            \item $T = \mathsf{SampleFullRank}(1^\lambda; F(K, id))$.
            \item Compute $w = T^{-1}(v)$ if $b = 0$; otherwise, compute $w = T^\mathsf{T}(v)$.
            \item Output $1$ if $w \in \cansbsp$ if $b = 0$ and if $w \in \cansbsp^\perp$ if $b = 1$. Otherwise, output $0$.
        \end{enumerate}
    \end{mdframed}

\item Sample $ptk \samp \rpke.\mathsf{SimulateTestKey}(pk)$.
\item Sample $\mathsf{OPReRand} \samp \io(\mathsf{PReRand})$ where $\mathsf{PReRand}$ is the following program\footnote{We note that if we make the stronger assumption of $\rpke$ with \emph{all-accepting simulatable testing keys}, we can actually remove the $\rpke.\mathsf{Test}$ line from the construction.}.
    
    \begin{mdframed}
        {\bf $\underline{\mathsf{PReRand}_{K}(id, s)}$}
        
        {\bf Hardcoded: $K, pk, ptk$}
        \begin{enumerate}[label=\arabic*.]
            \item Check if $\rpke.\mathsf{Test}(ptk, id) = 1$. If not, output $\perp$ and terminate.
            \item $id' = \rpke.\mathsf{ReRand}(pk, id; s)$.
            \item $T_1 = \mathsf{SampleFullRank}(1^\lambda; F(K, id))$.
            \item $T_2 = \mathsf{SampleFullRank}(1^\lambda; F(K, id'))$.
            \item Output $T_2\cdot T_1^{-1}$.
        \end{enumerate}
\end{mdframed}

\item Sample the proof $\pi \samp \mathsf{NIZK}.\mathsf{Prove}(crs_{nizk}, x = \mathsf{OPMem}, w = (K, r_{iO}))$.

\item Set $vk = (\mathsf{OPMem}, \mathsf{OPReRand}, \pi)$.
\item Set $mk = (K, pk)$.
\item Output $vk, mk$.
\end{enumerate}

\paragraph{\underline{$\mathsf{GenBanknote}(mk)$}}
\begin{enumerate}
    \item Parse $(K, pk) = mk$.
    \item Sample $ct \samp \rpke.\mathsf{Enc}(pk, 0^\lambda)$.
    \item $T = \mathsf{SampleFullRank}(1^\lambda; F(K, ct))$.
    \item Set $\ket{\$} = \sum_{v \in \cansbsp} \ket{T(v)}$.
    \item Output $ct, \ket{\$}$.
\end{enumerate}

\paragraph{\underline{$\mathsf{Verify}(crs, vk, \reg)$}}
\begin{enumerate}
    \item Parse $crs_{nizk} || pk = crs$ with $|pk| = q_2(\lambda).$
    \item Parse $(\mathsf{OPMem}, \mathsf{OPReRand}, \pi) = vk$.
    \item Check\footnote{We note that in reality, a user can just perform this verification only once on $vk$ and later on keep using $vk$.} if $\mathsf{NIZK}.\mathsf{Verify}(crs_{nizk}, x = \mathsf{OPMem}, \pi) = 1$. If not, output $0$ and terminate.
    \item Parse $(id, \reg') = \reg$.
    \item Run $\mathsf{OPMem}$ coherently on $id, \reg', 0$. Check if the output is $1$, and then rewind.
    \item Apply QFT to $\reg'$.
    \item Run $\mathsf{OPMem}$ coherently on $id, \reg', 1$. Check if the output is $1$, and then rewind.
    \item If any verification failed, output $0$ and terminate.
    \item Sample $s \samp \zo^{p_3(\lambda)}$.
    \item $T = \mathsf{OPReRand}(id, s)$.
    \item Compute $id' = \rpke.\mathsf{ReRand}(pk, id; s).$
    \item Apply the linear map $T: \F_2^\lambda \to \F_2^\lambda$ coherently to $\reg'$.
    \item Run $\mathsf{OPMem}$ coherently on $id', \reg', 0$. Check if the output is $1$, and then rewind.
    \item Apply QFT to $\reg'$.
    \item Run $\mathsf{OPMem}$ coherently on $id', \reg', 1$. Check if the output is $1$, and then rewind.
    \item If any verification failed, output $0$. Otherwise, output $1$.
\end{enumerate}

\begin{theorem}
    $\bank$ satisfies correctness and projectiveness.
\end{theorem}
The proof of this theorem follows similarly to \cref{sec:projective} and \cref{sec:correctness}.
\begin{theorem}
    $\bank$ satisfies untraceability (\cref{defn:untrac}).
\end{theorem}
See \cref{sec:proofuntrac} for the proof.

\begin{theorem}
    $\bank$ satisfies counterfeiting security (\cref{defn:bankcf}).
\end{theorem}
See \cref{sec:proofuntracunclon} for the proof.

\subsection{Proof of Untraceability}\label{sec:proofuntrac}
We will prove security through a sequence of hybrids, each of which is obtained by modifying the previous one.  Let $\adve$ be a QPT adversary.
\paragraph{$\underline{\hyb_0}$}: The original game $\qmuntrac{\adve}(1^\lambda)$.

\paragraph{$\underline{\hyb_1}$}: In the verification steps (Step 4 and 6 of the experiment), instead of performing $\bank.\mathsf{Verify}(vk, \reg)$, the challenger executes the code of $\bank.\mathsf{Verify}$ directly on $vk, \reg$. Further, instead of verifying the proof $\pi$ inside $vk$ twice (in Step 4 and 6), it performs it only once when $\adve$ outputs $vk$ at Step 3.

\paragraph{$\underline{\hyb_2}$}: After verifying the NIZK argument $\pi$ at Step 3, the challenger executes the following. It iterates over all strings $K, r$ and checks if $\mathsf{OPMem}$ (parsed from $vk$ as in $\bank.\mathsf{Verify}$) satisfies $\mathsf{OPMem} = \io(\mathsf{PMem}_K; r)$. If the challenger finds such a value, it sets $K^*, r^*$ to these value. If it cannot find such a value, it outputs $0$ and terminates.

Note that from this hybrid onwards, the challenger (and the experiments) will be exponential time.

\paragraph{$\underline{\hyb_3}$}: We change the way the challenger performs the banknote verification.
\paragraph{\underline{Verify Subroutine}$(\reg)$}
\begin{enumerate}
  \item Parse $crs_{nizk} || pk = crs$ with $|pk| = q_2(\lambda).$
    \item Parse $(\mathsf{OPMem}, \mathsf{OPReRand}, \pi) = vk$.
    \item Parse $(id, \reg') = \reg$.
    \item \textcolor{red}{Compute $T_{id} = \sampfrm{F(K^*, id)}.$}
    \item \textcolor{red}{Coherently run the following on $\reg'$: On input $v$, check if $T_{id}^{-1}(v) \in \cansbsp$.  Check if the output is $1$, and then rewind.}
    \item Apply QFT to $\reg'$.
    \item \textcolor{red}{Coherently run the following on $\reg'$: On input $v$, check if $T_{id}^{\mathsf{T}}(v) \in \cansbsp^\perp$.  Check if the output is $1$, and then rewind.}
    \item If any verification failed, output $0$ and terminate.
    \item Sample $s \samp \zo^{p_3(\lambda)}$.
    \item $T = \mathsf{OPReRand}(id, s)$.
    \item Compute $id' = \rpke.\mathsf{ReRand}(pk, id; s).$
    \item Apply the linear map $T: \F_2^\lambda \to \F_2^\lambda$ coherently to $\reg'$.
        \item \textcolor{red}{Compute $T_{id'} = \sampfrm{F(K^*, id')}.$}
    \item \textcolor{red}{Coherently run the following on $\reg'$: On input $v$, check if $T_{id'}^{-1}(v) \in \cansbsp$.  Check if the output is $1$, and then rewind.}
    \item Apply QFT to $\reg'$.
    \item \textcolor{red}{Coherently run the following on $\reg'$: On input $v$, check if $T_{id'}^{\mathsf{T}}(v) \in \cansbsp^\perp$.  Check if the output is $1$, and then rewind.}
    \item If any verification failed, output $0$. Otherwise, output $1$.
    \end{enumerate}

\paragraph{$\underline{\hyb_4}$}: Let $id_b$ the initial serial number and $id'_b$ be the rerandomized serial number ($id'_b = \rpke.\mathsf{ReRand}(pk ,id_b; s)$) obtained during verification subroutine for $\reg_b$. We change the challenge output behaviour of the challenger. Instead of sending $\reg_b$ to the adversary $\adve$ in the challenge case $b$, it instead sends $(id'_b, \ket{\psi_b})$ to the adversary where $\ket{\psi_b} = \sum_{v \in \cansbsp} \ket{T_{id'_b}(v)}$ and $T_{id'_b} = \sampfrm{F(K^*, id'_b)}.$

\begin{lemma}
    $\hyb_0 \equiv \hyb_1$.
\end{lemma}
\begin{proof}
    This is only a semantic change.
\end{proof}

\begin{lemma}
    $\hyb_1 \approx \hyb_2$.
\end{lemma}
\begin{proof}
    By the computational soundness of $\mathsf{NIZK}$, if it was the case that $\mathsf{OPMem} \not\in L$, then the proof $\pi$ produced by QPT adversary $\adve$ would pass verification with only negligible probability. Thus, if the challenger has not already terminated, we know that the proof $\pi$ passed and (except with negligibly small probability), $\mathsf{OPMem} \in L$. Hence, by definition of the language $L$, there exists $K, r$ such that $\mathsf{OPMem} = \io(\mathsf{PMem}_K; r)$, which the challenger will be able to find through exhaustive search.
\end{proof}

\begin{lemma}
    $\hyb_2 \equiv \hyb_3$.
\end{lemma}
\begin{proof}
    The change here is that the challenger executes the code $\mathsf{PMem}_{K^*}$ directly instead of using $\mathsf{OPMem}$. However, we have at this point that $\mathsf{OPMem} = \io(\mathsf{PMem}_{K^*}; r^*)$. By perfect correctness of $\io$, the result follows.
\end{proof}

\begin{lemma}
    $\hyb_3 \equiv \hyb_4$.
\end{lemma}
\begin{proof}
    As shown in \cref{sec:projective}, the verification procedure implemented here is projective. Thus, in $\hyb_3$ the challenger is already outputting $(id'_b, \ket{\psi_b})$ in the challenge phase, and there is no change between $\hyb_3$ and $\hyb_4$.
\end{proof}

\begin{lemma}
    $\Pr[\hyb_4 = 1] \leq \frac{1}{2} + \negl(\lambda)$.
\end{lemma}
\begin{proof}
    First note that the public key $pk$ is a truly random key, since it is taken from the CRS. Thus, by the \emph{statistical rerandomization security using truly random public key} property (\cref{defn:trurandstatrerand}), the ciphertexts $id'_0, id'_1$ (which are rerandomizations of $id_0, id_1$) are indistinguishable to the adversary, thus the challenge banknotes $(id'_0, \ket{\psi_0}), (id'_1, \ket{\psi_1})$ are also indistinguishable. Note that this still applies even though the adversary knows $K^*$ - since the security comes from indistinguishability of $id'_0, id'_1$. Finally, we note that while the challenger/experiments are exponential time, we are relying on \emph{statistical rerandomization security}, thus the security indeed applies.
\end{proof}

Combining the above shows $\Pr[\qmuntrac{\adve} = 1] \leq  \frac{1}{2} + \negl(\lambda)$, completing the proof.

\subsection{Proof of Unclonability (Counterfeiting) Security}\label{sec:proofuntracunclon}
We will prove security through a sequence of hybrids, each of which is obtained by modifying the previous one. The proof will follow similarly to \cref{sec:proofunclon} after $\hyb_3$. Let $\adve$ be a QPT adversary.

\paragraph{$\underline{\hyb_0}$}: The original game $\cfgame{\adve}(1^\lambda)$.

\paragraph{$\underline{\hyb_1}$}: At the beginning, the challenger samples $pk, ptk, sk \samp \rpke.\mathsf{Setup}(1^\lambda)$. Then, instead of sampling the CRS as $crs \samp \zo^{q(\lambda)}$, it samples $crs_{nizk} \samp \zo^{q_1(\lambda)}$ and sets $crs = crs_{nizk} || pk$.

\paragraph{$\underline{\hyb_2}$}: During $\bank.\mathsf{Setup}$, instead of sampling $ptk \samp \rpke.\mathsf{SimulateTestKey}(pk)$, the challenger instead uses $ptk$ it obtained from $\rpke.\mathsf{Setup}(1^\lambda)$.

\paragraph{$\underline{\hyb_3}$}: During $\bank.\mathsf{Setup}$, instead of sampling $crs_{nizk} \samp \zo^{q_1(\lambda)}$, the challenger instead simulates it as $(crs_{nizk}, st) \samp  \mathsf{NIZK}.\mathsf{SimulateCRS}(1^\lambda)$. Then, instead of honestly proving  $\pi \samp \mathsf{NIZK}.\mathsf{Prove}(crs_{nizk}, x = \mathsf{OPMem}, w = (K, r_{iO}))$, the challenger instead simulates the NIZK proof $\pi$ as $\pi \samp \mathsf{NIZK}.\mathsf{SimulateProof}(st, \mathsf{OPMem})$.

\paragraph{$\underline{\hyb_4}$}:  First, the challenger initializes a stateful counter $cnt = 1$ at the beginning. Further, we change the way the challenger samples banknotes. During minting, instead of sampling $ct \samp \rpke.\mathsf{Enc}(pk, 0^\lambda)$, it now samples $ct \samp \rpke.\mathsf{Enc}(pk, \textcolor{red}{cnt})$ and adds $1$ to $cnt$ afterwards.

\paragraph{$\underline{\hyb_5}$}: First, the challenger initializes a list $\mathsf{PL} = []$ at the beginning. Further, we change the verification subroutine of the challenger applied to adversary's forged banknotes. For a banknote $(id, \reg)$, after verifying the banknote, the challenger additionally performs the following check. It computes $pl = \rpke.\mathsf{Dec}(sk, id)$, adds $pl$ to the list $\mathsf{INDICES}$ and checks if $pl \in [k]$. If not, it outputs $0$ and terminates.

\paragraph{$\underline{\hyb_6}$}:  At the beginning of the game, the challenger samples a random index $i^*$. Also, we will set $id^*$ to the serial number of the $i^*$-th banknote produced by the challenger. Further, we require an additional condition for the adversary to win. At the end of the game, the challenger checks if $i^*$ appears at least twice in $\mathsf{INDICES}$. If not, the challenger outputs $0$ and the adversary loses.

\paragraph{$\underline{\hyb_7}$}: The challenger samples $K' \samp F.\mathsf{Setup}(1^\lambda)$ and a random full rank linear map $T^{*}: \F_2^\lambda \to \F_2^\lambda$ at the beginning of the game. We also compute $A^* = T^{*}(\cansbsp)$ and create the following function/program.
\begin{equation*}
    M_{K', \mathsf{CT}_{i^*}}(ct) = \begin{cases}
        \mathsf{SampleFullRank}(1^\lambda; F(K', id)), \text{ if } id \neq id^*\\
        I, \text{ if } id = id^*\\
    \end{cases}
\end{equation*}
Further, we now sample $\mathsf{OPMem}$ as $\mathsf{OPMem} \samp \io(\mathsf{PMem}')$ instead of sampling a random tape $r_{iO}$ and setting $\mathsf{OPMem} = \io(\mathsf{PMem}; r_{iO})$. We also change sampling of $\mathsf{OPReRand}$ to $\mathsf{OPReRand} \samp \io(\mathsf{PReRand}')$.

\begin{mdframed}
        {\bf $\underline{\mathsf{PMem}'_{K}(id, v, b)}$}
        
        {\bf Hardcoded: $K, \textcolor{red}{sk, isk, i^*, T^*, M_{K', \mathsf{CT}_{i^*}}}$}
        \begin{enumerate}[label=\arabic*.]
        \item \textcolor{red}{ Compute $pl = \rpke.\mathsf{Dec}(sk, id).$}
    \item \textcolor{red}{If $pl = i^*$, set $T =  M_{K', \mathsf{CT}_{i^*}}(id)\cdot T^*$.} Otherwise $T = \mathsf{SampleFullRank}(1^\lambda; F(K, id))$.
    
            \item Compute $w = T^{-1}(v)$ if $b = 0$; otherwise, compute $w = T^\mathsf{T}(v)$.
            \item Output $1$ if $w \in \cansbsp$ if $b = 0$ and if $w \in \cansbsp^\perp$ if $b = 1$. Otherwise, output $0$.
        \end{enumerate}
    \end{mdframed}

    \begin{mdframed}
        {\bf $\underline{\mathsf{PReRand}'_{K}(id, s)}$}
        
        {\bf Hardcoded: $K, ptk, \textcolor{red}{sk, isk, i^*, T^*, M_{K', \mathsf{CT}_{i^*}}}$}
        \begin{enumerate}[label=\arabic*.]
          \item Check if $\rpke.\mathsf{Test}(ptk, id) = 1$. Otherwise, output $\perp$ and terminate.
     \item \textcolor{red}{ Compute $pl = \rpke.\mathsf{Dec}(sk, id).$}
            \item $id' = \rpke.\mathsf{ReRand}(pk, id; s)$.
            \item \textcolor{red}{If $pl = i^*$, set $T_1 =  M_{K', \mathsf{CT}_{i^*}}(id)\cdot T^*$.} Otherwise, $T_1 = \mathsf{SampleFullRank}(1^\lambda; F(K, id))$.

            \item \textcolor{red}{ Compute $pl' = \rpke.\mathsf{Dec}(sk, id').$}
            \item \textcolor{red}{If $pl' = i^*$, set $T_2 =  M_{K', \mathsf{CT}_{i^*}}(id')\cdot T^*$.} Otherwise, $T_2 = \mathsf{SampleFullRank}(1^\lambda; F(K, id'))$.
            \item Output $T_2\cdot T_1^{-1}$.
        \end{enumerate}
\end{mdframed}
Finally, we modify the minting subroutine as follows.
\paragraph{\underline{Minting Subroutine($\reg$)}}
\begin{enumerate}
 \item Parse $(K, pk) = mk$.
 \item {Add $1$ to $cnt$.}
 
    \item Sample $ct \samp \rpke.\mathsf{Enc}(pk, 0^\lambda)$.
    \item \textcolor{red}{If $cnt = i^*$, set $id^* = ct$ and $\ket{\$} = \sum_{v \in A^*} \ket{v}$, and jump to the final step.}
    \item $T = \mathsf{SampleFullRank}(1^\lambda; F(K, ct))$.
    \item Set $\ket{\$} = \sum_{v \in \cansbsp} \ket{T(v)}$.
    \item Output $ct, \ket{\$}$.
    \end{enumerate}

\paragraph{$\underline{\hyb_{8}}$}: At the beginning of the game, after we sample $T^*$, we also sample $\mathsf{P}_0 \samp \io(A^*)$ and $\mathsf{P}_1 \samp \io((A^*)^\perp)$. Further, we now sample $\mathsf{OPMem}$ as $\mathsf{OPMem} \samp \io(\mathsf{PMem}'')$ .

\begin{mdframed}
        {\bf $\underline{\mathsf{PMem}''_{K}(id, v, b)}$}
        
        {\bf Hardcoded: $K, sk, isk, i^*, \textcolor{red}{\mathsf{P}_0, \mathsf{P}_1}$}
        \begin{enumerate}[label=\arabic*.]
        \item \textcolor{red}{ Compute $pl = \rpke.\mathsf{Dec}(sk, id).$}
    \item \textcolor{red}{If $pl = i^*$,}
    \begin{enumerate}[label=\arabic*.]
        \item\textcolor{red}{Set $T = M_{K', \mathsf{CT}_{i^*}}(id)$.}
        \item\textcolor{red}{ Compute $w = T^{-1}(v)$ if $b = 0$; otherwise, compute $w = (T^{-1})^\mathsf{T}(v)$.}
    \item\textcolor{red}{Output the output $\mathsf{P}_b(v)$ and terminate.}
    \end{enumerate}
    
    \item If $ipl_1 \neq i^*$,
    \begin{enumerate}[label=\arabic*.]
        \item Set $T = \mathsf{SampleFullRank}(1^\lambda; F(K, id))$.
        \item Compute $w = T^{-1}(v)$ if $b = 0$; otherwise, compute $w = (T)^\mathsf{T}(v)$.
        \item Output $1$ if $w \in \cansbsp$ if $b = 0$ and if $w \in \cansbsp^\perp$ if $b = 1$. Otherwise, output $0$.
    \end{enumerate}
        \end{enumerate}
    \end{mdframed}

\paragraph{$\underline{\hyb_{9}}$}: We now sample $\mathsf{OPReRand}$ as $\mathsf{OPReRand} \samp \io(\mathsf{PReRand}'')$ .
    \begin{mdframed}
        {\bf $\underline{\mathsf{PReRand}''_{K}(id, s)}$}
        
        {\bf Hardcoded: $K, ptk, {sk, isk, i^*}, M_{K', \mathsf{CT}_{i^*}}$}
        \begin{enumerate}[label=\arabic*.]
                \item Check if $\rpke.\mathsf{Test}(ptk, id) = 1$. Otherwise, output $\perp$ and terminate.
     
     \item { Compute $pl = \rpke.\mathsf{Dec}(sk, id).$}
            \item $id' = \rpke.\mathsf{ReRand}(pk, id; s)$.
            \item \textcolor{red}{If $pl = i^*$, set $T_1 =  M_{K', \mathsf{CT}_{i^*}}(id)$.} Otherwise, $T_1 = \mathsf{SampleFullRank}(1^\lambda; F(K, id))$.
            \item \textcolor{red}{If $pl = i^*$, set $T_2 =  M_{K', \mathsf{CT}_{i^*}}(id')$.} Otherwise, $T_2 = \mathsf{SampleFullRank}(1^\lambda; F(K, id'))$.
            \item Output $T_2\cdot T_1^{-1}$.
        \end{enumerate}
\end{mdframed}

\begin{lemma}
    $\hyb_0 \approx \hyb_1$
\end{lemma}
\begin{proof}
   The result follows from the pseudorandom public encryption key property of $\rpke$.
\end{proof}

\begin{lemma}
    $\hyb_1 \approx \hyb_2$
\end{lemma}
\begin{proof}
    The result follows from the indistinguishability of simulated testing keys and honest testing keys of $\rpke$.
\end{proof}

\begin{lemma}
    $\hyb_2 \approx \hyb_3$.
\end{lemma}
\begin{proof}
    The result follows by the computational zero knowledge security of $\mathsf{NIZK}$.
\end{proof}

\begin{lemma}
    $\hyb_3 \approx \hyb_4$.
\end{lemma}
\begin{proof}
    The result follows from the CPA-security of $\rpke$.
\end{proof}

\begin{lemma}
    $\hyb_4 \approx \hyb_5$.
\end{lemma}
\begin{proof}
    This result follows essentially from $0 \to 1$ unclonability of subspace states and by strong rerandomization correctness of $\rpke$. This proof follows similarly to \cref{lem:zerocopunlearn}.
\end{proof}

\begin{lemma}
    $\Pr[\hyb_6 = 1] \geq \frac{\Pr[\hyb_5 = 1]}{k}.$
\end{lemma}
\begin{proof}
    The adversary outputs $k+1$ forged banknotes, thus $|\mathsf{INDICES}| = k + 1$. However, all the values in $\mathsf{INDICES}$ are between $1$ and $k$. Thus, there exists $i^{**}$ such that $i^{**}$ appears at least twice in $\mathsf{INDICES}$, and the random $i^*$ satisfies $i^* = i^{**}$ with probability $1/k$.
\end{proof}

\begin{lemma}
    $\hyb_6 \approx \hyb_7$.
\end{lemma}
\begin{proof}
    Follows similarly to \cref{lem:switchtokprim}.
\end{proof}

\begin{lemma}
    $\hyb_7 \approx \hyb_8$.
\end{lemma}
\begin{proof}
    Observe that by correctness of the $\io$ used to sample $\mathsf{P}_1, \mathsf{P}_2$, the programs $\mathsf{PMem}', \mathsf{PMem}''$. Thus, the result follows by the security of $\io$ used to sample $\mathsf{OPMem}$.
\end{proof}

\begin{lemma}
    $\hyb_8 \approx \hyb_9$.
\end{lemma}
\begin{proof}
    This follows similarly to \cref{lem:unclonstrongrerand}: Note that due to strong rerandomization correctness of $\rpke$, there does not exist $id$ and $s$ such that $\rpke.\mathsf{Test}(tpk, id) = 1$ but $\rpke.\mathsf{Dec}(sk, id) \neq \rpke.\mathsf{Dec}(sk, \rpke.\mathsf{ReRand}(pk, id; s))$. Thus, in $\mathsf{PReRand}'$, we have that $pl$ and $pl'$ always have the same value. Thus, the tests $pl =^? i^*$ and $pl' =^? i^*$ will always go to the same case, and as a result the factor $T^*$ will be cancelled out when $T_2\cdot T_1^{-1}$ is computed. Hence, $\mathsf{PReRand}'$ and $\mathsf{PReRand}''$ have exactly the same functionality, and the result follows by security of $\io$.
\end{proof}

\begin{lemma}
    $\Pr[\hyb_9 = 1] \leq \negl(\lambda).$
\end{lemma}
\begin{proof}
The result follows by the $1\to 2$ unclonability of subspace states.
    The proof follows similarly to \cref{lem:unclonlasth}: we can show that from the forged banknotes of the adversary, we can convert two of them to $\ket{A^*}$, whereas the experiment can be simulated given a single copy $\ket{A^*}$. This means cloning $\ket{A^*}$, which is a contradiction by \cref{thm:onetotwo}.
\end{proof}

\section{Quantum Voting Schemes}\label{sec:qvs}
\subsection{Definitions}
In this section, we recall security notions for voting schemes. We will be working in the \emph{universal verifiability} setting: The users will cast their votes by posting them on a public bulletin board, and anyone can verify the validity of a cast vote using the public-key. Further, we will be working in the common (uniformly) random string model and will allow voting tokens to be quantum states (however, voting is done by simply posting a classical string on the bulletin board).

The first requirement we have is \emph{correctness}. In the quantum voting tokens setting, we extend the correctness notion to also apply against malicious voting tokens and verification keys. A voter will be assured that their vote will be correctly verified by the public.
\begin{definition}[Correctness against Malicious Keys and Tokens]    For any efficient algorithm $\mathcal{B}$ and candidate choice $c$, \begin{equation*}
        \Pr[b = 0 \vee b' = 1 : \begin{array}{c}
             \reg, vk \samp \mathcal{B}(vk, mk, tk) \\
             b \samp \qvs.\mathsf{VerifyVotingToken}(vk, \reg) \\
             vo \samp \qvs.\mathsf{Vote}(\reg, c) \\
             b' \samp \qvs.\mathsf{VerifyVote}(vk, vo, c)
        \end{array}] = 1.
    \end{equation*}
\end{definition}

We also introduce the notion of privacy against authorities, where privacy of votes holds against everyone, including the voting authority that creates the voting tokens. For our privacy requirement, we require that a vote from a token that the adversary has seen is indistinguishable from a vote from a fresh voting token. By a simple hybrid argument, this easily implies indistinguishability of votes of any two voters. In fact, we require privacy even for tokens created by the adversary.

\begin{definition}[Privacy]\label{defn:privqvs}
     Consider the following game between a challenger and a stateful adversary $\adve$.

    \paragraph{\underline{$\qvsuntrac{\adve}(1^\lambda)$}}
    \begin{enumerate}
    \item Sample $crs \samp \zo^{q(\lambda)}$.
    \item Submit $crs$ to the adversary $\adve$.
    \item Adversary $\adve$ outputs keys $vk, mk$, a register $\reg_0$ and candidate choice $c$.
    \item Run $\qvs.\mathsf{VerifyVotingToken}(vk, \reg_0)$. If it fails, output $0$ and terminate.
    \item Sample $\reg_1 \samp \qvs.\mathsf{GenVotingToken}(mk)$.
    \item Run $\qvs.\mathsf{VerifyVotingToken}(vk, \reg_1)$. If it fails, output $0$ and terminate.
    \item Sample $b \samp \zo$.
    \item Sample $vo \samp \qvs.\mathsf{Vote}(\reg_b, c)$.
    \item Submit $vo$ to the adversary $\adve$.
    \item Adversary $\adve$ outputs a bit $b'$.
    \item Output $1$ if and only if $b' = b$.
    \end{enumerate}

    We say that the quantum voting scheme $\qvs$ satisfies \emph{privacy} if for any QPT adversary $\adve$, we have 
    \begin{equation*}
        \Pr[\qvsuntrac{\adve}(1^\lambda) = 1] \leq \frac{1}{2} + \negl(\lambda).
    \end{equation*}
\end{definition}
We remark that our scheme will actually satisfy the stronger notion where the voting token (or a fresh token) is returned to the adversary (without being cast, as a quantum state), and the adversary still will not be able to win the privacy game. This trivially implies the definition above.

The final essential property we require is uniqueness: adversaries should not be able to more than once.
\begin{definition}[Uniqueness]\label{defn:qvsuniq}
     Consider the following game between a challenger and an adversary $\adve$.
    \paragraph{$\underline{\qvsuniq{\adve}(1^\lambda)}$}
\begin{enumerate}
    \item Sample $vk, mk \samp \qvs.\mathsf{Setup}(1^\lambda)$ and submit $vk$ to $\adve$.
    \item \underline{\textbf{Token Query Phase:}} For multiple rounds, $\adve$ queries for a voting token and the challenger executes $\regi{token} \samp \qvs.\mathsf{GenVotingToken}(mk)$ and submits $\regi{token}$ to the adversary. Let $k$ be the number of queries made by the adversary.
    \item $\adve$ outputs $k+1$ cast-votes $(vo_i = (c_i, s_i, \ell_i))_{i \in [k + 1]}$ where $c_i$ denotes the candidate choice and $\ell_i$ denotes the tag\footnote{A voter assigned serial number that is part of the cast-vote.}.
    \item The challenger tests $\qvs.\mathsf{Verify}(vk, (c_i, s_i, \ell_i)) = 1$ for $i \in [k + 1]$. It also checks that all tags $\ell_i$ are unique. It outputs $1$ if all the tests pass; otherwise, it outputs $0$.
\end{enumerate}
We say that the quantum voting scheme $\qvs$ satisfies \emph{uniqueness security} if for any QPT adversary $\adve$,

\begin{equation*}
    \Pr[\cfgame{\adve}(1^\lambda) = 1] \leq \negl(\lambda).
\end{equation*}
\end{definition}
\subsection{Construction}
In this section, we give our universally verifiable quantum voting scheme construction, in the common random string model. In our scheme, verifying a voting token automatically rerandomizes it. Our scheme will be similar to our untraceable quantum money scheme (\cref{sec:consuntrac}), with the main difference being that it will contain $2\cdot\lambda$ money states in voting tokens.

We assume the existence of the same primitives as in \cref{sec:consuntrac}.

\paragraph{\underline{$\qvs.\mathsf{Setup}(crs)$}}
\begin{enumerate}
\item Parse $crs_{nizk} || pk = crs$ with $|crs_{nizk}| = q_1(\lambda)$ and $|pk| = q_2(\lambda).$
    \item Sample $K \samp F.\mathsf{Setup}(1^\lambda)$.
    \item Sample random tape $r_{iO}$ and compute $\mathsf{OPMem} = \io(\mathsf{PMem}; r_{iO})$ where $\mathsf{PMem}$ is the following program.
    
\begin{mdframed}
        {\bf $\underline{\mathsf{PMem}_{K}(id, (v_i)_{i \in [2\cdot \lambda]}, b)}$}
        
        {\bf Hardcoded: $K$}
        \begin{enumerate}[label=\arabic*.]
            \item Parse $r_1 || \dots || r_{2\lambda} = F(K, id)$ into $2\cdot \lambda$ equal size strings.
            \item For $i \in [2\cdot \lambda]$,
            \begin{enumerate}[label=\arabic*.]
                \item Set $T_i = \mathsf{SampleFullRank}(1^\lambda; r_i)$.
                \item Compute $w_i = T^{-1}(v)$ if $b_i = 0$; otherwise, compute $w_i = T^\mathsf{T}(v_i)$.
                \item Check if $w_i \in \cansbsp$ if $b_i = 0$ and if $w_i \in \cansbsp^\perp$ if $b_i = 1$.
            \end{enumerate}
            \item If all the checks pass, output $1$. Otherwise, output $0$.
        \end{enumerate}
    \end{mdframed}

\item Sample $ptk \samp \rpke.\mathsf{SimulateTestKey}(pk)$.
\item Sample $\mathsf{OPReRand} \samp \io(\mathsf{PReRand})$ where $\mathsf{PReRand}$ is the following program\footnote{We note that if we make the stronger assumption of $\rpke$ with \emph{all-accepting  simulatable testing keys}, we can actually remove the $\rpke.\mathsf{Test}$ line from the construction.}.

    \begin{mdframed}
        {\bf $\underline{\mathsf{PReRand}_{K}(id, s)}$}
        
        {\bf Hardcoded: $K, pk, ptk$}
        \begin{enumerate}[label=\arabic*.]
        \item Check if $\rpke.\mathsf{Test}(ptk, id) = 1$. If not, output $\perp$ and terminate.
            
        \item $id' = \rpke.\mathsf{ReRand}(pk, id; s)$.
        \item Parse $r_1 || \dots || r_{2\lambda} = F(K, id)$ into $2\cdot \lambda$ equal size strings.
        \item Parse $r'_1 || \dots || r'_{2\lambda} = F(K, id')$ into $2\cdot \lambda$ equal size strings.
            \item For $i \in [2\cdot \lambda]$,
        
\begin{enumerate}[label=\arabic*.]
                \item Set $T_{i,1} = \mathsf{SampleFullRank}(1^\lambda; r_i)$.
                    \item Set $T_{i,2} = \mathsf{SampleFullRank}(1^\lambda; r'_i)$.
                \item Set $T_{i, \mathsf{out}} = T_{i,2}\cdot T_{i,1}^{-1}$.
            \end{enumerate}   
            
            \item Output $(T_{i, \mathsf{out}})_{i \in [2\cdot\lambda]}$
        \end{enumerate}
\end{mdframed}

\item Sample the proof $\pi \samp \mathsf{NIZK}.\mathsf{Prove}(crs_{nizk}, x = \mathsf{OPMem}, w = (K, r_{iO}))$.

\item Set $vk = (\mathsf{OPMem}, \mathsf{OPReRand}, \pi)$.
\item Set $mk = (K, pk)$.
\item Output $vk, mk$.
\end{enumerate}

\paragraph{\underline{$\qvs.\mathsf{GenVotingToken}(mk)$}}
\begin{enumerate}
    \item Parse $(K, pk) = mk$.
    \item Sample $ct \samp \rpke.\mathsf{Enc}(pk, 0^\lambda)$.
        \item Parse $r_1 || \dots || r_{2\lambda} = F(K, ct)$ into $2\cdot \lambda$ equal size strings.
         \item For $i \in [2\cdot \lambda]$,
\begin{enumerate}[label=\arabic*.]
\item $T_i = \mathsf{SampleFullRank}(1^\lambda; r_i)$.
\item Set $\ket{\$_i} = \sum_{v \in \cansbsp} \ket{T_i(v)}$.
\end{enumerate}

    \item Output $ct, \bigotimes_{i \in [2\cdot \lambda]}\ket{\$_i}$.
\end{enumerate}

\paragraph{\underline{$\qvs.\mathsf{VerifyVotingToken}(crs, vk, \reg)$}}
\begin{enumerate}
    \item Parse $crs_{nizk} || pk = crs$ with $|pk| = q_2(\lambda).$
    \item Parse $(\mathsf{OPMem}, \mathsf{OPReRand}, \pi) = vk$.
    \item Check if $\mathsf{NIZK}.\mathsf{Verify}(crs_{nizk}, x = \mathsf{OPMem}, \pi) = 1$. If not, output $0$ and terminate.
    \item Parse $(id, \reg') = \reg$.
    \item Run $\mathsf{OPMem}$ coherently on $id, \reg', 0^{2\cdot \lambda}$. Check if the output is $1$, and then rewind.
    \item Apply QFT to $\reg'$.
    \item Run $\mathsf{OPMem}$ coherently on $id, \reg', 1^{2\cdot \lambda}$. Check if the output is $1$, and then rewind.
    \item If any verification failed, output $0$ and terminate.
    \item Sample $s \samp \zo^{p_3(\lambda)}$.
    \item $(T_1, \dots, T_{2\cdot\lambda}) = \mathsf{OPReRand}(id, s)$.
    \item Compute $id' = \rpke.\mathsf{ReRand}(pk, id; s).$
    \item For $i \in [2\cdot \lambda]$, apply the linear map $T_i: \F_2^\lambda \to \F_2^\lambda$ coherently to the $i$-th part of the register $\reg'$.
    \item Run $\mathsf{OPMem}$ coherently on $id, \reg', 0^{2\cdot \lambda}$. Check if the output is $1$, and then rewind.
    \item Apply QFT to $\reg'$.
    \item Run $\mathsf{OPMem}$ coherently on $id', \reg', 1^{2\cdot \lambda}$. Check if the output is $1$, and then rewind.
    \item If any verification failed, output $0$. Otherwise, output $1$.
\end{enumerate}

\paragraph{\underline{$\qvs.\mathsf{Vote}(\reg, c)$}}
\begin{enumerate}
    \item Sample $r \in \zo^\lambda$.
    \item Parse $(id, (\reg'_i)_{i \in [2\cdot\lambda]}) = \reg$.
    \item For $i \in [2\cdot\lambda]$: If $(c || r)_i = 0$, measure $\reg'_i$ in the computational basis to obtain $v_i$. If $(c || r)_i = 1$, measure $\reg'_i$ in the Hadamard basis (i.e. QFT-and-measure) to obtain $v_i$.
    \item Output $c, (id, (v_i)_{i \in [2\cdot\lambda]}), r$.
\end{enumerate}

\paragraph{\underline{$\qvs.\mathsf{VerifyCastVote}(\reg, c, vo)$}}
\begin{enumerate}
    \item Parse $(\mathsf{OPMem}, \mathsf{OPReRand}, \pi) = vk$.
    \item Parse $(id, (v_i)_{i \in [2\cdot\lambda]}, r) = vo$.
    \item Check if $\mathsf{OPMem}(id, (v_i)_{i \in [2\cdot\lambda]}, c || r)$.
\end{enumerate}

\begin{theorem}
    $\qvs$ satisfies correctness against malicious tokens and keys.
\end{theorem}
\begin{proof}
    By the same argument as in the projectiveness proof of our quantum money schemes \cref{sec:projective}, we can show that the token verification algorithm of our scheme projects onto  a perfect token. Thus, the result follows.
\end{proof}

\begin{theorem}
    $\qvs$ satisfies privacy.
\end{theorem}
\begin{proof}
    This follows similarly to the untraceability of our quantum money scheme (\cref{sec:proofuntrac}).
\end{proof}
\begin{theorem}
    $\qvs$ satisfies uniqueness security.
\end{theorem}
\begin{proof}
    Through the same proof (generalized to multiple subspace states) as the unclonability proof of our quantum money scheme (\cref{sec:proofuntracunclon}), we can reduce the uniqueness security of our scheme to the direct product hardness of subspace states (\cref{thm:dphard}), since creating $k+1$ valid votes with all different tags will amount to adversary finding $v, w$ with $v \in A$ and $w \in A^\perp$.
\end{proof}

\section{Acknowledgements}
Alper Çakan was supported by the following grants of Vipul Goyal: NSF award 1916939, DARPA SIEVE program, a gift from Ripple, a DoE NETL award, a JP Morgan Faculty Fellowship, a PNC center for financial services innovation award, and a Cylab seed funding award.

\bibliographystyle{alpha}
\bibliography{refs}
\appendix

\section{Additional Definitions}\label{sec:extradef}
\begin{definition}[Counterfeiting Security]\label{defn:bankcf}
    Consider the following game between the challenger and an adversary $\adve$.
    \paragraph{$\underline{\cfgame{\adve}(1^\lambda)}$}
\begin{enumerate}
    \item Sample $vk, mk, tk \samp \bank.\mathsf{Setup}(1^\lambda)$ and submit $vk, tk$ to $\adve$.
    \item \underline{\textbf{Banknote Query Phase:}} For multiple rounds, $\adve$ queries for a banknote by sending a tag $t$. For each query, the challenger executes $\regi{bn} \samp \bank.\mathsf{GenBanknote}(mk, t)$ and submits $\regi{bn}$ to the adversary. Let $k$ be the number of queries made by the adversary.
    \item $\adve$ outputs a $(k+1)$-partite register $(\reg_i)_{i \in [k + 1]}$.
    \item The challenger tests $\bank.\mathsf{Verify}(vk, \reg_i) = 1$ for $i \in [k + 1]$. It outputs $1$ if all the tests pass; otherwise, it outputs $0$.
\end{enumerate}
We say that the quantum money scheme $\bank$ satisfies \emph{counterfeiting security} if for any QPT adversary $\adve$,

\begin{equation*}
    \Pr[\cfgame{\adve}(1^\lambda) = 1] \leq \negl(\lambda).
\end{equation*}
\end{definition}

\subsection{Anonymity}\label{defn:anonold}
In this section, we introduce various anonymity notions for public-key quantum money. Similar to previous work \cite{amr20,bs21}, our anonymity definition either randomly permutes or does not touch the banknote registers, according to a random bit; and the adversary needs to predict which case it is. However, our definition is significantly stronger than previous work: The adversary receives the minting key at the beginning whereas in previous definitions it received it at the final step of the security game. This also means that unlike previous definitions, we do not need to have a banknote query phase: the adversary can mint its own challenge banknotes.
\begin{definition}[Anonymity]
    Consider the following game between a challenger and an adversary $\adve$.

    \paragraph{\underline{$\qmanon{\adve}(1^\lambda)$}}
    \begin{enumerate}
    \item Sample $vk, mk, tk \samp \bank.\mathsf{Setup}(1^\lambda)$.
    \item Submit $\textcolor{blue}{vk, mk}$ to $\adve$.
    \item Adversary $\adve$ outputs a value $k$ and a (possibly entangled) $k$-partite register $(\reg_i)_{i \in [k]}$.
    \item Run $\bank.\mathsf{Verify}(vk, \reg_i)$ for $i \in [k]$. If any of them fails, output $0$ and terminate.
    \item Run $\bank.\mathsf{ReRandomize}(vk, \reg_i)$ for $i \in [k]$.
    \item Sample a permutation $\pi: [k] \to [k]$ and a bit $b$.
    \item If $b = 0$, submit $(\reg_i)_{i \in [k]}$; otherwise, submit $(\reg_{\pi(i)})_{i \in [k]}$ to the adversary  $\adve$.
    \item Adversary $\adve$ outputs a bit $b'$.
    \item Output $1$ if and only if $b' = b$.
    \end{enumerate}

    We say that the quantum money scheme $\bank$ satisfies \emph{anonymity} if for any QPT adversary $\adve$, we have 
    \begin{equation*}
        \Pr[\qmanon{\adve}(1^\lambda) = 1] \leq \frac{1}{2} + \negl(\lambda).
    \end{equation*}
\end{definition}

\section{Omitted Proofs}
\subsection{Proof of \cref{lem:switchtokprim}}\label{sec:switchtokprim}
This lemma follows by a (standard) PRF puncturing argument where we create hybrids for all strings $id$. In each hybrid $j$, we will switch to the new serial number - linear map association for all $id$ satisfying $id < j$. To go from $j$ to $j+1$, we do a small number of subhybrids where we first puncture the PRF key at $j$, replace the associated $T$ with a truly random sample. We also puncture $K'$. This allows us to switch over to the new mapping, and then depuncture $K'$ and $K$.
\subsection{Proof of \cref{lem:zerocopunlearn}}\label{sec:zerocopunlearn}
To prove this, first we do the same argument as in \cref{sec:switchtokprim} to switch to a completely new association between serial numbers and linear maps $T$ for any serial number $id$ that decrypts to a value not in $[k]$. Then, by the same argument as in the rest of the unclonability proof and as in \cref{lem:unclonlasth}, we can show that if the adversary forges any valid banknote whose serial number decrypts to a value $\not\in [k]$, we can extract a subspace state $\ket{A^*}$, with using only $\io(A^*), \io((A^*)^\perp)$. This is a contradiction to the unlearnability (or $0 \to 1$ unclonability) of subspaces.

\section{Remark on Representing the Public Key}\label{appn:remark}
 A remark here is in order. In our RPKE constructions, we will have that the public key is a matrix in $\Z_q^{m \times n}$, not a binary string. However, we can always represent such a matrix entrywise, where each number $\Z_q$ can be publicly and canonically represented as a binary string of length $\ceil{\log_2(q)}$ (any extra binary strings can be rolled over and assigned to unique values in $\Z_q$), and we can efficiently go back and forth between any value in $\Z_q$ and any binary string in $\zo^{\ceil{\log_2(q)}}$. Thus, without loss of generality we can assume that the public key is a binary string.

\end{document}